\documentclass[a4paper,11pt]{article}
\usepackage[lined,commentsnumbered,ruled,noend,
boxed]{algorithm2e}
\usepackage{graphicx,amssymb,amsmath,amsthm,textcomp}
\usepackage{authblk}
\usepackage{tikz}
\input{psfig.sty}
\newtheorem{theorem}{Theorem}
\newtheorem{lemma}{Lemma}

\newtheorem{defn}{Definition}
\usepackage{color}

\newtheorem{observation}{Observation}

\newcommand{\remove}[1]{}

\usepackage{wrapfig}
\newtheorem{fact}{Fact}
\newenvironment{Proof}{\noindent {\bf Proof:\,\ }}{\hfill\mbox{\ $\Box $}\smallskip}

\title{Approximation algorithms for the two-center problem of convex polygon}
\author{Sanjib Sadhu$^1$, Sasanka Roy$^2$, Soumen Nandi$^2$, Anil Maheswari$^3$\footnote{{Supported by NSERC}}, ~and \\Subhas C.
Nandy$^2$\footnote{{Part  of the work was done while this author was visiting Carleton University in Fall 2015.}}}
\date{$^1$ Dept. of Computer Science and Engineering, National Institute of Technology Durgapur, India\\
$^2$ Indian Statistical Institute, Kolkata, India\\
$^3$ School of Computer Science, Carleton University, Ottawa, Canada}

\begin{document}
\maketitle

\begin{abstract}Given a convex polygon $P$ with $n$ vertices, 
the two-center problem is to find two congruent closed 
disks of minimum radius such that they completely cover 
$P$. 
We propose an algorithm for this problem in the streaming 
setup, where the input stream is the vertices of the polygon 
in clockwise order. It produces a radius $r$ satisfying 
$r\leq2r_{opt}$ using $O(1)$ space, where $r_{opt}$ is the 
optimum solution.
Next, we show that in non-streaming setup, we can improve
the approximation factor by $r\leq 1.84 r_{opt}$, maintaining 
the time complexity of the algorithm to $O(n)$, and using 
$O(1)$ extra space in addition to the space required for 
storing the input.
\end{abstract}


\renewcommand*{\thefootnote}{\fnsymbol{footnote}}

\section{Introduction}\vspace{-0.1in}
Covering a geometric object (e.g., a point set or a polygon)  by 
disks has drawn a lot of interest to the researchers due to its 
several applications, for example, base station placement in 
mobile network, facility location in city planning, etc. There 
are mainly two variations of the disk cover problem, namely 
standard version and discrete version, depending on the position 
of the centers of the disks to be placed. In standard version, the position 
of centers of disks are anywhere on the plane, whereas in the 
discrete version, the center of the disks must be on some 
specified points, also given as input. The objective of a 
$k$-center problem for a given set of points $S$ in a metric 
space is to find out $k$ points (also called {\it centers}) 
$c_1,c_2, \ldots, c_k$ in the underlying space so that the 
largest distance of a point $p\in S$ from its nearest center 
$c\in\{c_1,c_2,\ldots,c_k\}$ is minimized. In other words, 
in $k$-center problem we want to cover a set of points using $k$ 
congruent balls of minimum radius. In this paper, we consider 
the standard two-center problem for a convex polygon $P$ in 
the $L_2$ metric, where the objective is to identify centers 
of two congruent closed disks whose union completely covers the 
polygon $P$ and their (common) radius $r$ is minimum. As stated 
by Kim and Shin~\cite{kim}, the major difference between 
the {\it two-center problem for a convex polygon $P$ } and the 
{\it two-center problem for a point set $S$} are (1) points
covered by the two disks in the former problem are {\it in 
convex positions} (instead of arbitrary positions), and (2) 
the union of two disks should also cover the edges of the 
polygon $P$. The feature (1) indicates the problem is easier 
than the standard two-center problem for points, but feature 
(2) says that it might be more difficult.

\subsection{Related work}
 The $k$-center problem, where $S$ is 
a set of points in a Euclidean plane and the distance function 
is the $L_2$ metric, is NP-complete for any dimension $d\geq 2$~\cite{marchetti}. 
Therefore it makes sense to study the $k$-center problem for 
small (fixed) values of $k$ (\cite{regis1,timothy,david,john,
JaromczykK94,KatzKS00,Sharir97}) and to search for efficient 
approximation algorithms and heuristics for the general version 
(\cite{JaromczykK94},~\cite{Plensik}). Hershberger~\cite{john} proposed an 
$O(n^2\log n)$ time algorithm for the standard version of the two-center
problem for the $n$-points in plane. Sharir~\cite{Sharir97} improved the time 
complexity of the problem to $O(n\log^9n)$. Eppstein~\cite{david} 
proposed a randomized algorithm with expected time complexity 
$O(n\log^2 n)$. Later, Chan~\cite{timothy} proposed two algorithms 
for this problem. The first one is a randomized 
algorithm that runs in $O(n\log^2 n)$ time with high probability, 
and the second one is a deterministic algorithm that runs in 
$O(n\log^2 n(\log \log n)^2)$ time. The discrete version of the 
two-center problem for a point set was solved by Agarwal et al. 
\cite{pankaj} in $O(n^{4/3}\log^5n)$ time. The standard and 
discrete versions of the two-center problem for a convex polygon 
$P$ was first solved by Kim and Shin~\cite{kim} in 
$O(n\log^3n\log\log n)$  and $O(n\log^2n)$ time respectively,  
where $n$ is the number of vertices of $P$. Recently Vigan 
\cite{ivo} proposed the problem of covering a simple polygon by 
$k$ geodesic disks whose centers lie inside $P$. Here, the 
geodesic distance between a pair of points $s$ and $t$ inside 
the polygon is the length of the shortest $s-t$ path inside $P$. 
He showed that the maximum radius among these $k$ geodesic disks 
is at most twice as large as that of an optimal solution, and 
the time complexity of the proposed algorithm is $O(k^2(n + k) 
\log(n + k))$. The algorithm proposed by Vigan~\cite{ivo}, if applied for convex polygon, 
the approximation factor remains unaltered\footnote{However, this algorithm is in the non-streaming setup. In non-streaming
setup, we have better result.}.
There exists a heuristic to cover a convex region
by $k$ congruent disks of minimum radii~\cite{goutam}. However,
to the best of our knowledge
there are no linear time approximation algorithm for the $k$-center problem
of a convex polygon, where $k\geq 2$.

In the streaming model, McCutchen et al.~\cite{McCutchenK08} and
Guha~\cite{Guha09} have designed a ($2+\epsilon$)-approximation 
algorithm for the $k$-center problem of a point set in $\mathbb{R}^d$ 
using $O(\frac{kd}{\epsilon}\log(\frac{1}{\epsilon}))$ space. For 
the 1-center problem, Agarwal and Sharathkumar~\cite{AgarwalS15} suggested 
a $((1+\sqrt{3})/2+\epsilon)$-factor approximation algorithm 
using $O(\frac{d}{\epsilon^3}\log(\frac{1}{\epsilon}))$ space.
The approximation factor was later improved to $1.22$ by Chan 
and Pathak~\cite{ChanP14}. Recently, Kim and Ahn~\cite{KimA15} 
proposed a ($1.8+\epsilon$)-approximation algorithm for the two-center 
problem of a point set in $\mathbb{R}^2$. It uses $O(\frac{d}{\epsilon})$
space and update time where insertion and deletion of the points in the 
set are allowable. To the best of our knowledge, there is no 
approximation result for the two-center problem for a convex polygon
under the streaming model.

\subsection{Our result}
We propose a $2$-factor approximation algorithm for the 
two-center problem of a convex polygon in streaming setup. 
Here, the vertices of the input polygon is read in clockwise 
manner and the execution needs $O(n)$ time using $O(1)$ space. 
Next we show that if the restriction on streaming model is 
relaxed, then we can improve the approximation factor to 1.84 
maintaining the time complexity to $O(n)$ and using $O(1)$ 
extra space apart from the space required for storing the input.
We have observed the fact that if two disks cover a convex polygon $P$,
then they must also cover a ``line segment'' or a ``triangle'' lying inside that
polygon $P$. This fact has been used in our work to analyze the approximation factor of the radius of disks.

\subsection{Notations and terminologies used}
Throughout the paper we use the following notations. 
The line segment joining any two points $p$ and $q$ is denoted by
$\overline{pq}$ and its length is denoted by $|{pq}|$.
The $x$- and $y$-coordinate of a point $p$ are 
denoted by $x(p)$ and $y(p)$ respectively. The ``{\it horizontal distance}'' 
between a pair of points $p$ and $q$ is $|x(p)-x(q)|$ (the 
absolute difference between their
$x$-coordinates).
Similarly, the ``{\it vertical  distance}'' between a pair of points
$p$ and $q$ is $|y(p)-y(q)|$. The notation $s\in \overline{pq}$
implies that the point $t$ lies on $\overline{pq}$. 
We will 
use $\triangle$, $\Box $ and $\Diamond  $ to represent triangle, 
axis-parallel rectangle, and quadrilateral of arbitrary orientation of edges respectively. 

\subsection{Organization of the paper}
In this paper, the Section~\ref{problem_alg} describes the algorithm for
two-center problem of a convex polygon in streaming setup along with 
the detailed analysis of the approximation factor. Section~\ref{non_str}
 discusses the same problem under non-streaming model and a linear time algorithm
 is proposed along with a detailed discussion on the analysis of approximation factor. 
 Finally we conclude in section~\ref{conc} with future work.

\section{Two-center problem for convex polygon under streaming model} 
\label{problem_alg}
In this section, we first describe the streaming algorithm for the problem in subsection~\ref{str_alg}.
Then in subsection~\ref{lower_bound}, we discuss about the type of lower bounds of the optimal radius of the disks
followed by the interesting characteristic of the problem in subsection~\ref{approx_stream} which shows that only
quadrilaterals, triangles are to be studied instead of all convex polygons for the approximation factor.
Subsection~\ref{analysis_sm} will show the detailed analysis of the approximation factor.
 
\subsection{Proposed algorithm}
\label{str_alg}
Under the streaming data model, the algorithm has only a limited amount of working space.
So it cannot store all the input items received so far. In this model, 
the input data is read only once in a single pass. It does not require the entire data set to be stored in memory.
In the streaming setup, the vertices of the convex polygon $P$ 
arrives in order one at a time. In a linear scan among the vertices of $P$,
we can identify the four vertices $a$, $b$, $c$ and $d$ of the polygon $P$ 
with minimum $x$-, maximum $x$-, minimum $y$- and maximum $y$-coordinate 
respectively as shown in Figure~\ref{sm1}(a). This needs $O(1)$ scalar 
locations. Let $\mathcal{R}=\Box efgh$ be  an axis-parallel rectangle whose four sides passes through the vertices $a, b, c$ and $d$ of the convex polygon $P$, where $a\in \overline{gh}$,
$b\in \overline{he}$, $c\in \overline{ef}$ 
and $d\in\overline{fg}$.
%\begin{wrapfigure}{rh}{0.5\textwidth}
\begin{figure}
\centering
\includegraphics[scale=0.65]{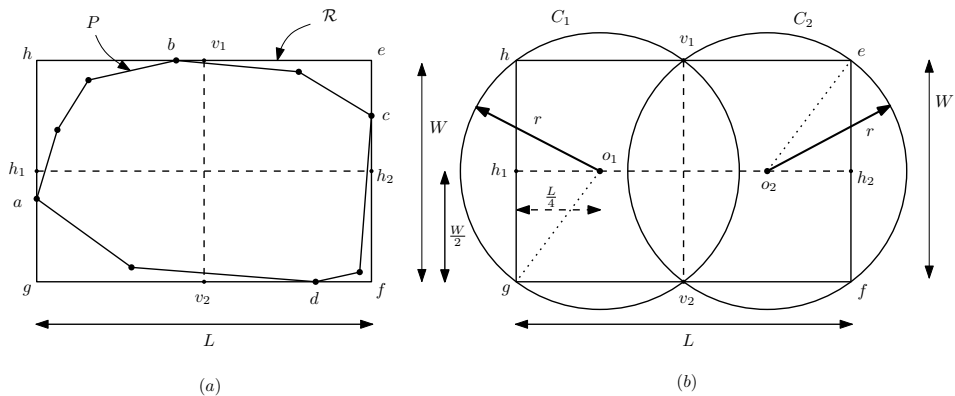}
\caption{(a) Covering rectangle ${\cal R}=\Box efgh$ for the convex polygon $P$.  (b) Rectangle ${\cal R}$ of size $L\times W$
covered by two disks $C_1$ and $C_2$}
\label{sm1}
\end{figure}
%\end{wrapfigure}
The length and width of rectangle $\mathcal{R}$ are  $L=|x(c)-x(a)|$ and  
$W=|y(b)-y(d)|$ respectively. We split $\mathcal{R}$ into two equal parts 
${\cal R}_1$ and ${\cal R}_2$ by a vertical line $\overline{v_1v_2}$, 
where $v_1 \in \overline{eh}$ and $v_2 \in \overline{fg}$ (see 
Figure~\ref{sm1}(a)). Finally, compute two congruent disks $C_1$ 
and $C_2$ of minimum radii circumscribing ${\cal R}_1$ and ${\cal R}_2$ 
respectively (see Figure~\ref{sm1}(b)). The output of our algorithm 
is $r$, the radii of $C_1$ (resp. $C_2$). {\it Since the two disks 
cover the rectangle $\mathcal{R}$ together, they must also cover the 
polygon $P$ lying inside $\mathcal{R}$}. 
For an axis-parallel 
rectangle ${\mathcal R} = \Box {efgh}$ of length $L$ and width $W$ 
(where $0 < W \leq L$) covering the polygon $P$, the value of $r$ (as shown in Figure~\ref{sm1}(b))
computed by our algorithm is
\begin{equation}
\label{eq1}
r=\sqrt{\left(\frac{L}{4}\right)^2+\left(\frac{W}{2}\right)^2}
 =\frac{1}{4}\sqrt{L^2+4W^2}
\end{equation} 

The time complexity of 
our algorithm, determined mainly by identification of the
four vertices $a$, $b$, $c$ and $d$ during the streaming input of the vertices of $P$, 
takes $O(n)$ time, where $n$ is the size of the input. 
	
 Let $r$ be the radius of the two congruent disks $C_1$ and $C_2$ for enclosing $P$, 
 returned by our algorithm. If $r_{opt}$ is the minimum radius of the two congruent disks that 
cover $P$, then the 
approximation factor 
of our algorithm is $\alpha=\frac{r}{r_{opt}}$. We now 
propose a lower bound $\rho$ of $r_{opt}$, which 
suggests an upper bound $\frac{r}{\rho}$ of $\alpha$, i.e. $\alpha\leq\frac{r}{\rho}$.

\subsection{Lower bound $\rho$ of  
$r_{opt}$}
\label{lower_bound}

\begin{defn} \label{exact-cover}
A convex polygon $P$ is said to be {\em {\bf exactly covered}}
by an axis-parallel rectangle ${\cal R}$, if $P \cap {\cal R}=P$ and 
each of the four side of ${\cal R}$
contain at least one vertex of $P$.
\end{defn}

\begin{defn} \label{subset}
A convex polygon $P_1$ is said to be a {\em {\bf subpolygon}}
of a convex polygon $P_2$, if the set of vertices of $P_1$ are subset of
the vertices of $P_2$ and this is denoted by $P_1 \subseteq P_2$. 
\end{defn}

The Figure~\ref{sm1}(a) shows that the convex polygon $P$ is
exactly covered by the rectangle $\cal R$ (Definition \ref{exact-cover}) and the 
quadrilateral $\Diamond abcd \subseteq P$ (Definition \ref{subset}).

Now, to have a better estimate of the approximation factor, we need a 
lower bound of $r_{opt}$, which is as large as possible. The following  
observations give us an idea of choosing two types of lower bound of $r_{opt}$.

\begin{observation}
\label{covering}
The two disks 
whose union covers the convex polygon $P$, must also 
cover a convex polygon which is a subpolygon of $P$.
\end{observation}

Thus, the lower bound 
of the radii of the two disks for covering a  quadrilateral $\Diamond abcd$, where $\Diamond abcd \subseteq P$, is 
also a lower bound for the radius of the two-center problem for  the convex polygon $P$.

\begin{observation}
\label{lower_bound_1}
Let $\mathcal{L}$ be the longest line segment within a 
quadrilateral $\Diamond abcd$ inside $P$. The two disks 
whose union covers the convex polygon $P$, must also 
cover the line segment $\cal L$ because $\Diamond abcd 
\subseteq P$.
\end{observation}

From Observation~\ref{lower_bound_1}, we conclude that $\rho\geq \frac{{\cal L}}{4}$. 
Moreover, the length of the line segment $\cal L$ can be at 
most $D$, the diameter of the convex polygon $P$.

\begin{observation}\label{obsrv1}
Let $\Delta$ be a triangle inside the polygon $P$. 
If a pair of disks 
$C_1$ and $C_2$ completely cover $P$, they must also cover the 
triangle $\Delta$. Again, if a pair of disks 
$C_1$ and $C_2$ cover a triangle $\Delta$, one of them must 
fully cover one of the edges of $\Delta$. 
\end{observation}

Thus, a  lower 
bound $\rho$ of $r_{opt}$ is half of the length of the smallest 
edge of a triangle inside $P$ (Observation~\ref{obsrv1}). In order to tighten the lower 
bound we find a triangle $\Delta$ inside $P$ whose smallest 
edge is as large as possible. We use $\ell$ to denote the 
smallest edge of $\Delta$. We also use $\ell$ to denote the length of $\ell$.
Thus, {\it ${\ell}/{2}$ is a lower 
bound for $\rho$}. 

Note that, in our analysis $\Delta$ may not always be the 
triangle whose smallest side is of maximum length among 
all triangles inscribed in $P$. We try to find a triangle 
$\Delta$ inscribed in $P$ such that the length of its 
smallest side $\ell$ has a closed form expression in terms 
of the length ($L$) and width ($W$) of the rectangle $\cal R$ 
covering $P$. This helps us to establish an upper bound
on the 
approximation factor $\alpha$ of our algorithm.

\subsection{Characterization of the problem}
\label{approx_stream}

The upper bound of the  
approximation factor $\alpha$ for the two-center problem for the polygon $P$ is $\alpha
\leq\frac{r}{({\ell}/{2})}$ or, $\alpha\leq \frac{r}{\left(|{\cal L}|/4\right)}$ 
depending on the type of lower bound used. In order to have a 
worst case estimate of the approximation factor, at first we fix $r$ 
(or in other words both $L$ and $W$ of the rectangle ${\cal R}$). Now, there are different convex polygons
exactly covered by the same rectangle $\cal R$, and 
the lower bound of optimal radius for each such polygon are possibly
different. Thus  
in order to have a worst estimate of the upper bound for the 
approximation factor $\alpha$, we choose the polygon $P$ inside $\cal R$ 
for which the lower bound ($\rho$) of $r_{opt}$ is minimum among
all possible polygons inside ${\cal R}$. The following 
observation gives us an intuition for choosing quadrilaterals
and triangles instead of inspecting all 
possible polygons exactly covered by the rectangle $\cal R$.

\begin{observation}\label{approx_1}
Let $P$ be a convex polygon which is exactly 
covered by an axis-parallel rectangle $\cal R$ of length $L$ and
width $W$ {\em($W\leq L$)}. Let $\Pi$ be a subpolygon of $P$ {\em($\Pi \subseteq P$)} so that
$\Pi$ is also {\em exactly
covered} by the same axis-parallel rectangle $\cal R$.
Then the upper bound of the approximation factor $\alpha$ of our algorithm for 
polygon $P$ will be smaller than {\em(or equal to)} that for polygon $\Pi$. 
\end{observation}

\begin{Proof}
Follows from the Observation~\ref{covering} that the lower bound
$\rho$ ($(\ell/2)$ or $(|{\cal L}|/4)$) of the optimal radius $r_{opt}$
for polygon $\Pi$ will be less than that for polygon $P$
(because of the fact that any triangle $\Delta$ in $\Pi$
or any line segment $\cal L$ in $\Pi$
also lies inside $P$).
\end{Proof}

Observation~\ref{approx_1} says that in order to measure the upper bound of the approximation 
factor of our algorithm for a given convex polygon $P$, one should choose a quadrilateral 
$\Diamond abcd$ as a subpolygon $\Pi$ of $P$ (i.e. $\Pi=\Diamond abcd\subseteq P$) where both $P$ and 
$\Pi=\Diamond abcd$ are exactly covered by the same rectangle $\cal R$. The reason for choosing 
the quadrilateral $\Diamond abcd$ as subpolygon $\Pi$ of $P$ is that quadrilateral is the minimal 
convex polygon (``minimal'' in the sense that ``there exists no subpolygon of the quadrilateral 
$\Diamond abcd$ which is exactly covered by the same rectangle $\cal R$''). 
From now onwards, we will use $\Pi$ to denote ``{\em a subpolygon of $P$ such that both $P$
and $\Pi$ are exactly covered by the same rectangle $\cal R$}''.
It needs to mention 
that, we may have two degenerate cases, (i) if a vertex $p$ 
of the given convex polygon $P$ coincides with a vertex of $\cal R$, then
the minimal subpolygon $\Pi$ of $P$ will be a triangle with one 
of its vertex at $p$, and (ii) if (maximum-$x$, maximum-$y$), and (minimum-$x$, minimum-$y$)
coordinates correspond to two vertices, say $p$ and $q$, of the given convex polygon $P$ 
(i.e. any two non-adjacent corners, say $e$ and $g$, of the rectangle $\cal R$ coincides with 
these two vertices $p$ and $q$), then we need to consider diagonal $\overline{eg}$ as a 
subpolygon $\Pi$ (with area  zero). Now note that, whatever be the shape of a convex polygon $P$ 
that is exactly covered by rectangle $\cal R$, we always obtain a subpolygon $\Pi$ as a quadrilateral $\Diamond abcd
\subseteq P$ (including degeneracies). The observation~\ref{approx_1} says that the approximation 
factor for this given convex polygon $P$ will be bounded above by that of its subpolygon $\Pi=\Diamond abcd$. 
Therefore we will concentrate on all possible quadrilaterals inside $\cal R$ 
rather than studying convex $n$-gons with $n\geq 5$.
Now, each such quadrilaterals have different lower bound of optimal radius.
The minimum of these lower bounds for $r_{opt}$ among all possible quadrilaterals will be used to compute the
upper bound of the approximation factor for an arbitrary convex polygon which is 
exactly covered by the rectangle $\cal R$.

In our streaming model we have stored only the four vertices $a$, $b$, $c$ and $d$ of the convex 
polygon $P$ and we find out either a triangle $\Delta$ (as defined in earlier section),
or the longest line segment $\cal L$ inside the $\Diamond abcd$
instead of searching them inside $P$ and the approximation factor thus obtained gives an upper bound for the same in $P$.

In the next subsection, we perform an exhaustive case analysis and 
finally present a flowchart in Figure  \ref{flowchart_stream} to justify the following result. 

\begin{theorem}
The approximation factor $\alpha$ of the two-center problem for a convex polygon $P$ in the streaming model is $2$.
\end{theorem}
\begin{proof}
Follows from Lemma \ref{lll1}, \ref{lll2} and \ref{lll3}, stated in the next subsection. 
\end{proof}

\subsection{Analysis of approximation factor}
\label{analysis_sm}

Let, $h_1$, $v_1$, $h_2$ and $v_2$ be the mid-points of $\overline{gh}$, 
$\overline{he}$, $\overline{ef}$ and $\overline{fg}$ respectively, and 
$|{fg}|=L$ and $|{ef}|=W$ ($L\geq W$ as shown in Figure~\ref{sm1}(b)). 
Surely, $L\leq D$ (the diameter of the
polygon $P$). We study, in detail, the case when $\Pi$
is a quadrilateral. We also discuss the two degenrate cases, namely, 
(i) $\Pi$ is a triangle and (ii) $\Pi$ is a diagonal $\overline{eg}$ of ${\cal R}=\Box efgh$.

\subsubsection{$\Pi$ is a quadrilateral $\Diamond abcd$}

We consider the following two cases separately.  

{\bf Case I: $0<\frac{W}{L} \leq \frac{\sqrt{3}}{2}$} 
 
 One of the diagonals of $\Diamond abcd$ (e.g. $\overline{ac}$ in
 Figure~\ref{sm1}(a))  must be at least of length $L$ and the two congruent 
 disks must cover this diagonal. Thus, we have $|{\cal L}|\geq L$, and hence  $\rho\geq\frac{L}{4}$, implying  
 $\alpha=\frac{r}{\rho}\leq\frac{\frac{1}{4}\sqrt{L^2+4W^2}}{L/4}=\sqrt{1+4\left(\frac{W}{L}\right)^2}$.
 Since $\left(\frac{W}{L}\right)\leq \frac{\sqrt{3}}{2}$, we have $\alpha\leq 2$.

\begin{figure}[htbp]
\centering
\includegraphics[width=\linewidth]{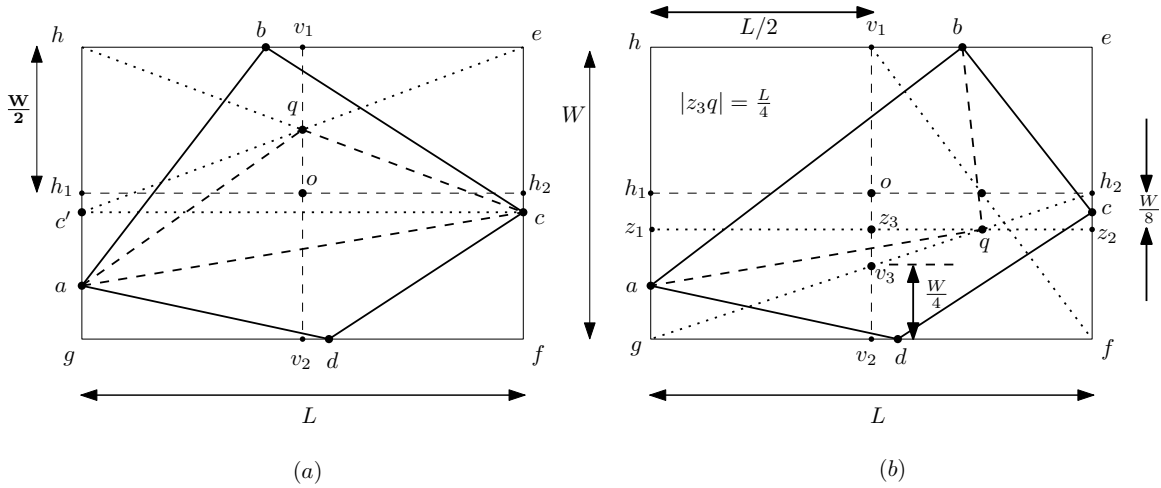}
\caption{Demonstration of Observation~\ref{obsh1h2}: ($a$) Case ($i$) and
($b$) Case ($ii$).}
\label{sm2}
\end{figure}

{\bf Case II: $\frac{\sqrt{3}}{2} <\frac{W}{L} \leq 1$} 

Before studying this case, we show the following two important observations:

\begin{observation}
\label{obsh1h2}
If both the vertices $a$ and $c$ lie at the same side of $\overline{h_1h_2}$, the approximation factor $\alpha$ will be  2.
\end{observation}
\begin{Proof}
Without loss of generality, assume that both $a$ and $c$ lie below $\overline{h_1h_2}$,
and  $y(a) < y(c)$ as shown in Figure~\ref{sm2}. 
Depending on the position of $b$ on the edge $\overline{eh}$, we consider the following two cases :
\subsection*{Case (i): $b$ lies on $\overline{hv_1}$}
Refer to Figure~\ref{sm2}(a).
Choose a point $c' \in \overline{gh}$ such that 
$y(c) =y(c')$. Let $q$ be the point of intersection of $\overline{ch}$ and $\overline{c'e}$. 
Whatever be the position of the vertex $b$ on $\overline{hv_1}$, the point $q$ must lie always inside 
the quadrilateral $\Diamond abcd$. Hence, the triangle $\triangle aqc$ will also lie inside
the $\Diamond abcd$.
Now, as $|ce|\geq \frac{W}{2}$, we have 
$|cq|=|c'q|=\frac{1}{2}|c'e| \geq \frac{1}{2}|h_1e|= \frac{1}{2}\sqrt{L^2+\left(W/2\right)^2}$. 
  Here, we choose the triangle $\triangle aqc$. Since the point $a$ lies below $c'$,
  we have $|aq|\geq |c'q|=|cq|$. Also, $|ac|\geq L>\frac{1}{2}\sqrt{L^2+\left(W/2\right)^2}$. Therefore, the smallest side $\ell$ of the triangle
  $\triangle aqc$ will be at least  $\frac{1}{2}\sqrt{L^2+\left({W}/{2}\right)^2}$.
  Hence, 
  $\alpha=\frac{r}{\ell/2}=\sqrt{\frac{L^2+4W^2}{L^2+\left({W}/{2}\right)^2}}=
  2\sqrt{\frac{L^2+4W^2}{4L^2+W^2}}\leq 2$ (since $W\leq L$).
  
\subsection*{Case (ii): $b$ lies on $\overline{v_1e}$}
Refer to Figure~\ref{sm2}(b). Consider a horizontal line segment $\overline{z_1z_2}$ below $\overline{h_1h_2}$ at a 
distance of $\frac{W}{8}$. This segment $\overline{z_1z_2}$ intersects $\overline{gh_2}$ at point $q$. Now,
since $c\in \overline{h_2f}$, $d\in \overline{gf}$ and $b\in \overline{v_1e}$, the point
$q$ must lie always within $\Diamond abcd$. Hence the triangle $\triangle abq$ will also lie inside
$\Diamond abcd$. Here, we choose this triangle $\triangle abq$. Now, $\overline{v_1v_2}$ intersect 
$\overline{h_1h_2}$, $\overline{z_1z_2}$
and $\overline{gh_2}$ at the points $o$, $z_3$ and $v_3$ respectively, where $|v_3v_2|=|v_3o|=\frac{1}{2}|h_2f|=\frac{W}{4}$.
From the similar triangles $\triangle gz_1q$ and $\triangle v_3z_3q$, we have
\begin{equation}
\label{eqn_sm}
\frac{|z_1g|}{|z_3v_3|}=\frac{|z_1q|}{|z_3q|} 
\end{equation}
Now, $|z_1g|=W-\left(\frac{W}{2}+\frac{W}{8}\right)=\frac{3W}{8}$,
$|z_3v_3|=|v_3o|-|z_3o|=\left(\frac{W}{4}-\frac{W}{8}\right)=\frac{W}{8}$ and $|z_1q|=|z_1z_3|+|z_3q|=\frac{L}{2}+|z_3q|$.
Hence, the Equation~\ref{eqn_sm} gives $|z_3q|=\frac{L}{4}$. Now, in $\triangle abq$, $|aq|\geq |z_1q|=\left(|z_1z_3|+|z_3q|\right)=\frac{3L}{4}$,
$|bq|\geq \left(\frac{W}{2}+\frac{W}{8}\right)=\frac{5W}{8}$ and $|ab|\geq \sqrt{\left(\frac{L}{2}\right)^2+ \left(\frac{W}{2}\right)^2}$.
Since $L\geq W$, we have $\sqrt{\left(\frac{L}{2}\right)^2+ \left(\frac{W}{2}\right)^2}\geq \sqrt{\left(\frac{W}{2}\right)^2+ \left(\frac{W}{2}\right)^2}
=0.707W>\frac{5W}{8}$. Also, $\frac{3L}{4}\geq\frac{5W}{8}$, because $L\geq W$. Therefore, the smallest side 
$\ell$ of the $\triangle abq$ must be at least $\frac{5W}{8}$. Therefore, the approximation factor $\alpha$
will be given by $\alpha=\frac{r}{\ell/2}=\frac{\frac{1}{4}\sqrt{L^2+4W^2}}{5W/16}=\frac{4}{5}\sqrt{4+\left(L/W\right)^2}
< \frac{4}{5}\sqrt{4+\left(4/3\right)}$ since $\frac{L}{W} < \frac{2}{\sqrt{3}}$
($\because \frac{\sqrt{3}}{2}<\frac{W}{L}\leq 1$). Thus, we have $\alpha=\frac{16}{5\sqrt{3}}<2$.
\end{Proof}

\begin{figure}[htbp]
\centering
\includegraphics[width=\linewidth]{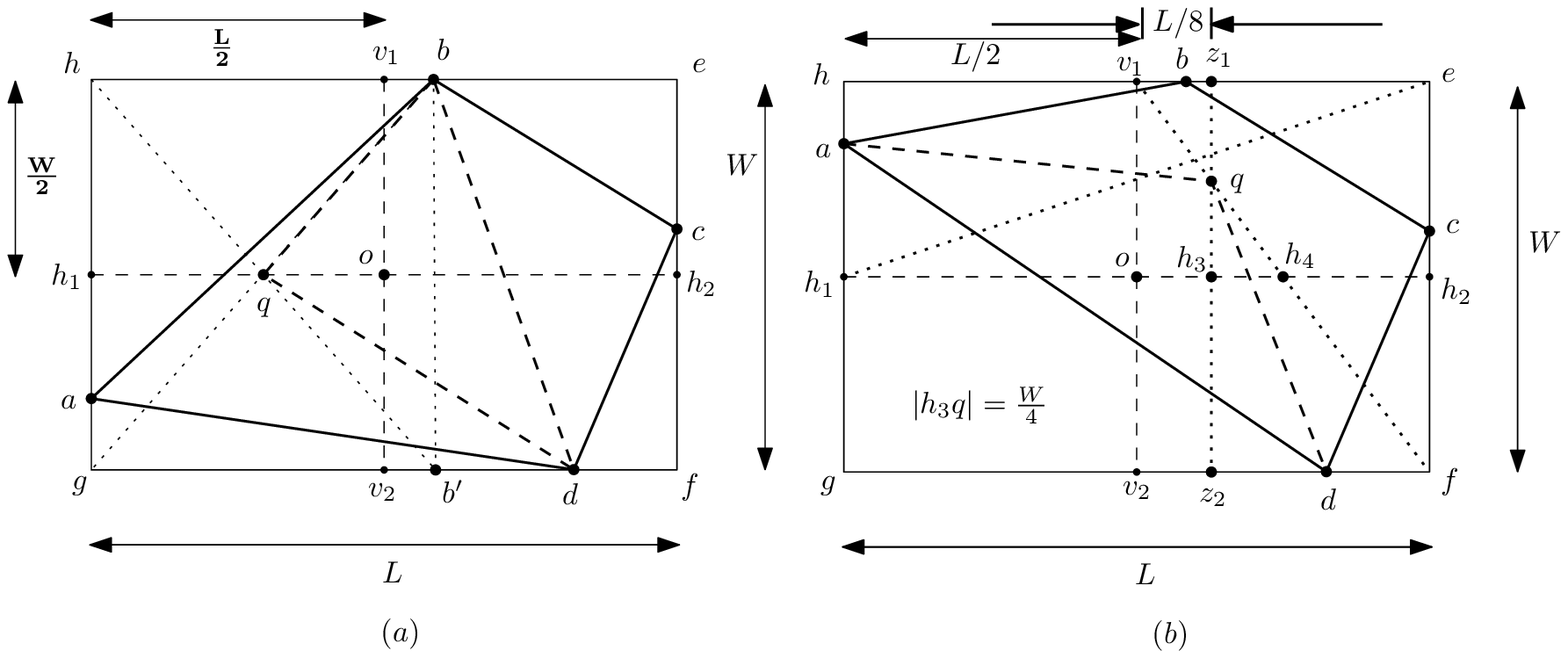}
\caption{Demonstration of Observation~\ref{obsv1v2}: ($a$) Case ($i$) and
($b$) Case ($ii$).}
\label{sm3}
\end{figure}

\begin{observation}
\label{obsv1v2}
If both the vertices $b$ and $d$ lie at the same side of $\overline{v_1v_2}$, the approximation factor $\alpha$ will be  2.
\end{observation}
\begin{Proof}
Without loss of generality, assume that $b$ and $d$ lie to the right of $\overline{v_2v_1}$ and $x(d) \geq x(b)$ as shown in Figure~\ref{sm3}. 
Depending on the position of $a$ on the edge $\overline{gh}$, we study the following two cases:
\subsection*{Case (i): $a$ lies on $\overline{gh_1}$}
Refer to Figure~\ref{sm3}(a). Let $b'$ be a point on edge $\overline{fg}$ with $x(b')=x(b)$. The line segments $\overline{bg}$ and 
$\overline{b'h}$ intersect at $q$. Whatever be the position of $a$ on $\overline{gh_1}$,  the point $q$
must lie inside the quadrilateral $\Diamond abcd$. Hence, the triangle $\triangle bqd$, as shown in
Figure~\ref{sm3}(a), must also lie inside $\Diamond abcd$. Here, we choose the triangle $\triangle bqd$.
Now, since $|bh|\geq \frac{L}{2}$, we have 
$ |bq|=\frac{1}{2}|bg|\geq \frac{1}{2}|v_1g|= \frac{1}{2}\sqrt{W^2+\left(\frac{L}{2}\right)^2}$.
Surely, $|qd|\geq|b'q|=|bq|$.
Also, $|bd|\geq |{bb'}| = W$.  
Thus the smallest side $\ell$ of triangle $\triangle bqd$ will be either $\overline{bq}$ or $\overline{bd}$. Now,  
\begin{description}
\item[$\bullet$] if $|bd|\leq|bq|$, we have $\ell=|bd| \geq W$. Therefore, 
$\alpha=\frac{r}{\ell/2}\leq\frac{1}{2}\sqrt{4+(L/W)^2}$. Since 
$\frac{W}{L}>\frac{\sqrt{3}}{2}$, we have $\alpha\leq \frac{1}{2}\sqrt{4+(4/3)}<2$.

\item[$\bullet$] if $|bd|>|bq|$, we have $\ell= |bq|\geq\frac{1}{2}\sqrt{W^2+\left(L/2\right)^2}$.
 Hence $\alpha=\frac{r}{\left(\ell/2\right)}\leq\sqrt{\frac{L^2+4W^2}{W^2+\left(L/2\right)^2}}=2$.
\end{description}
\subsection*{Case (ii): $a$ lies on $\overline{h_1h}$}
Refer to Figure~\ref{sm3}(b). Consider a vertical line segment $\overline{z_1z_2}$ 
to the right of $\overline{v_2v_1}$ at a distance of $\frac{L}{8}$.
This segment $\overline{z_1z_2}$ intersect $\overline{v_1f}$ at $q$. Now,
since $a\in \overline{h_1h}$, $b\in \overline{v_1e}$ and $c\in \overline{ef}$, the point
$q$ must always lie within $\Diamond abcd$. Hence the triangle $\triangle adq$ will also lie inside
$\Diamond abcd$. We choose this triangle $\triangle adq$. Now, $\overline{h_1h_2}$ intersect $\overline{z_1z_2}$
and $\overline{v_1f}$ at the points $h_3$ and $h_4$ respectively, where $|h_4h_2|=|h_4o|=\frac{1}{2}|v_1e|=\frac{L}{4}$, where $o$
is  the mid-point of $\overline{h_1h_2}$.
From the similar triangles $\triangle z_2fq$ and $\triangle h_3h_4q$, we have
\begin{equation}
\label{eqn_sm_1}
\frac{|z_2q|}{|z_2f|}=\frac{|h_3q|}{|h_3h_4|} 
\end{equation}

Now, $|z_2q|=|h_3q|+\frac{W}{2}$, $|z_2f|=L-\left(\frac{L}{2}+\frac{L}{8}\right)=\frac{3L}{8}$
and  $|h_3h_4|=\left(\frac{L}{4}-\frac{L}{8}\right)=\frac{L}{8}$. Thus, Equation~\ref{eqn_sm_1}
gives $|h_3q|=\frac{W}{4}$. Now, in $\triangle adq$, $|aq|\geq |h_1h_3|=\frac{5L}{8}$,
$|qd|\geq |z_2q|=\frac{3W}{4}$ and $|ad|\geq \sqrt{\left(\frac{L}{2}\right)^2+ \left(\frac{W}{2}\right)^2}$.
Any one of the three sides of $\triangle aqd$ can be smallest side $\ell$. Now,  
\begin{description}
\item[$\bullet$] if $\ell=|qd| \geq \frac{3W}{4}$, we have 
$\alpha=\frac{r}{\ell/2}\leq \frac{\frac{1}{4}\sqrt{L^2+4W^2}}{\frac{3W}{8}}=\frac{2}{3}\sqrt{4+(L/W)^2}$. Since 
$\left(\frac{L}{W}\right)\leq\frac{2}{\sqrt{3}}$, we have $\alpha\leq \frac{2}{3}\sqrt{4+(4/3)}=\frac{8}{3\sqrt{3}}<2$.

\item[$\bullet$] if $\ell= |aq|\geq \frac{5L}{8}$, we have
  $\alpha=\frac{r}{\left(\ell/2\right)}\leq\frac{4}{5}\sqrt{1+4\left(W/L\right)^2}\leq \frac{4}{\sqrt{5}}<2$, since $\left(W/L\right)\leq 1$.
  
\item[$\bullet$] if $\ell=|ad|\geq \sqrt{\left(\frac{L}{2}\right)^2+ \left(\frac{W}{2}\right)^2}$, we have 
$\alpha=\frac{r}{\left(\ell/2\right)}\leq\sqrt{\frac{L^2+4W^2}{L^2+W^2}}=\sqrt{1+\frac{3}{1+\left(L/W\right)^2}}$.\\
Now, since $\left(L/W\right)\geq 1$, the approximation factor $\alpha$ is given by  $\alpha\leq\sqrt{1+\frac{3}{2}}<2$.
\end{description}
\vspace{-0.1in}
\end{Proof}

Observations \ref{obsh1h2} and \ref{obsv1v2} say that if either ``$a$ and $c$ lie in the same side of $\overline{h_1h_2}$''
or, ``$b$ and $d$ lie in the same side of $\overline{v_1v_2}$'', or both, then $\alpha\leq 2$.
Thus, it remains to analyze the 
case where both ``$a$ and $c$ lie on the opposite sides of $\overline{h_1h_2}$'', and  
``$b$ and $d$ lie on the opposite
sides of $\overline{v_1v_2}$''.
Without loss of generality,  we assume that $a\in \overline{gh_1}$ and $c\in \overline{eh_2}$. Now, depending on the positions of $b$ and
$d$, we need to consider the following two cases. 
\subsection*{Case (A): The vertices $b$ and $d$ lie at the left and right of $\overline{v_2v_1}$ respectively.}

Here, $b\in \overline{v_1h}$ and $d\in \overline{v_2f}$ (see Figure~\ref{sm4}). We need to consider the 
following two cases depending on the length of the edge $\overline{cd}$ of 
$\Diamond abcd$. 

{\bf Case A.1 ~ $|cd|\geq\sqrt{\left(L/4\right)^2+\left(W/2\right)^2}$} 

Refer to Figure~\ref{sm4}(a). 
The point $q$, determined by the intersection of $\overline{gv_1}$ and 
$\overline{h_1h_2}$, must lie inside $\Diamond abcd$. This is because of the 
fact that $b$ lies to the left of $v_1$ on $\overline{v_1h}$ and $a$ lies above 
$g$ on $\overline{gh_1}$. Here, $|oq|=(L/4)$, where $o$ is the mid point of 
$\overline{h_1h_2}$. We choose $\triangle cdq$, where 
$|dq|\geq |v_2q|= \sqrt{\left(L/4\right)^2+\left(W/2\right)^2} $ and 
$|cq|\geq |h_2q|=\left({3L}/{4}\right)$. As $|cd|\geq \sqrt{\left(L/4\right)^2+\left(W/2\right)^2}$, in $\triangle 
cdq$ we have $\ell\geq \sqrt{\left(L/4\right)^2+\left(W/2\right)^2}$, 
and hence $\alpha=\frac{r}{\left(\ell/2\right)}\leq \frac{\frac{1}{4}\sqrt{L^2+4W^2}}{\frac{1}{8}\sqrt{L^2+4W^2}}=2$.

\begin{figure}[htbp]\vspace{-0.1in}
\begin{minipage}[b]{0.5\linewidth}
\centering
\includegraphics[width=\linewidth]{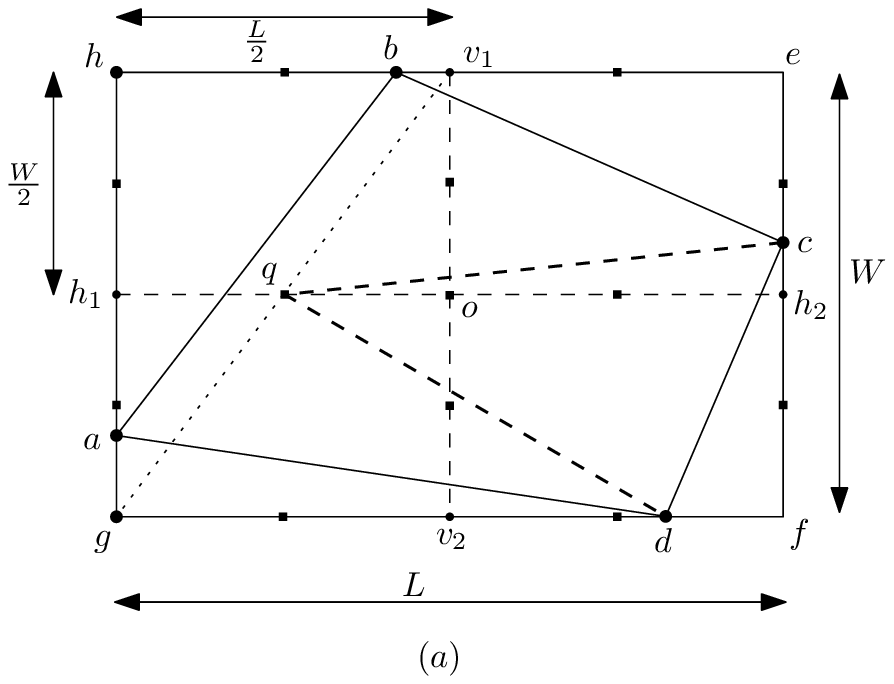}
%\caption{(a)}
%\label{sm4}
\end{minipage} %
\hspace{0.5cm}
\begin{minipage}[b]{0.5\linewidth}
\centering
\includegraphics[width=\linewidth]{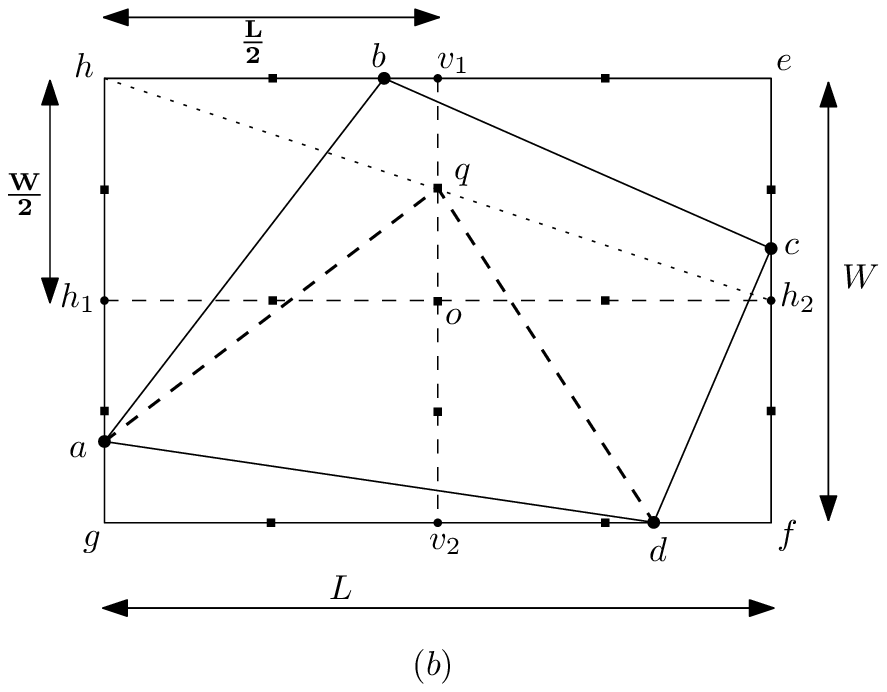}
%\caption{(b)}
%\label{sm5}
\end{minipage} %
\caption{Demonstration of Case A: (a) $|cd|\geq \sqrt{\left(L/4\right)^2+\left(W/2\right)^2}$, and
(b) $|cd|< \sqrt{\left(L/4\right)^2+\left(W/2\right)^2}$}
\label{sm4}
\vspace{-0.1in}
\end{figure}

{\bf Case A.2~  $|cd|<\sqrt{\left(L/4\right)^2+\left(W/2\right)^2}$} 

Refer to Figure~\ref{sm4}(b).  
The necessary condition for this case is that the point $d$ must lie at the right of the 
mid-point of $\overline{v_2f}$, i.e. $|df|< L/4$. Therefore, 
$|ad|\geq |gd|= \frac{L}{2}+\frac{L}{4}=\frac{3L}{4}$. 
Consider the point $q$ which is determined by the intersection of $\overline{v_1v_2}$ and $\overline{hh_2}$.
 Hence, $|oq|=\frac{W}{4}$. In this case, the point $q$
 must lie within the quadrilateral $\Diamond abcd$ because of the constraints that 
 $c\in \overline{eh_2}$ cannot lie below $h_2$, and $b$ lies on $\overline{hv_1}$.
 We choose the triangle $\triangle adq$. In this triangle, 
 $|dq|\geq |do|=\sqrt{\left(L/4\right)^2+\left(W/2\right)^2} $ and 
 $|aq|\geq |h_1q|=\sqrt{\left(L/2\right)^2+\left(W/4\right)^2}$. 
Now, $|ad|\geq \frac{3L}{4}\geq \sqrt{\left(L/4\right)^2+\left(W/2\right)^2}$ because $L\geq W$.
 Also note that, $\sqrt{\left(L/2\right)^2+\left(W/4\right)^2}\geq \sqrt{\left(L/4\right)^2+\left(W/2\right)^2}$\footnote{
 Since $ L\geq W$, we have $\frac{3L^2}{16}\geq\frac{3W^2}{16}\implies \left(L/2\right)^2+\left(W/4\right)^2 \geq \left(L/4\right)^2+\left(W/2\right)^2$ }.
 Hence, in $\triangle adq$, we have $\ell\geq \sqrt{\left(L/4\right)^2+\left(W/2\right)^2}$,
 and $\alpha=\frac{r}{\left(\ell/2\right)}\leq \frac{\frac{1}{4}\sqrt{L^2+4W^2}}{\frac{1}{8}\sqrt{L^2+4W^2}}=2$.

\vspace{-0.15in}
\subsection*{Case (B): The vertices $b$ and $d$ lie at the right and left of $\overline{v_2v_1}$ respectively.}

Here, $b\in \overline{ev_1}$ and $d\in\overline{gv_2}$ as shown in Figure~\ref{sm6a},~\ref{sm7a} and~\ref{sm8}.
This case is again divided into two sub-cases depending on the ``vertical distance'' between $a$ and $c$.

{\bf Case B.1:~  $|y(c)-y(a)|\geq \frac{W}{2}$}

Refer to Figure~\ref{sm6a}. Let $n_1$ and $n_2$ be 
the mid-points of $\overline{gh_1}$ and $\overline{eh_2}$ respectively.
 Observe that, if $a\in \overline{h_1n_1}$ then 
 $c\notin\overline{h_2n_2}$ and vice-versa, otherwise the given condition
 $|y(c)-y(a)|\geq \frac{W}{2}$ will become invalid. 
Since $|y(c)-y(a)|\geq(W/2)$, we also have 
$|{ac}|\geq |n_1n_2|=\sqrt{L^2+(W/2)^2}$. Since the longest line segment $\cal L$
inside the $\Diamond abcd$ is at least $|{ac}|$, we have the lower bound $\rho={\cal L}/4\geq |{ac}|/4$. 
Hence, the approximation factor $\alpha\leq\frac{r}{|ac|/4}\leq\sqrt{\frac{L^2+4W^2}{L^2+(W/2)^2}}=
2\sqrt{\frac{L^2+4W^2}{4L^2+W^2}}\leq 2$ (since $L\geq W $).

\begin{figure}[htbp]
\begin{minipage}[b]{0.5\linewidth}
\centering
\includegraphics[scale=0.8]{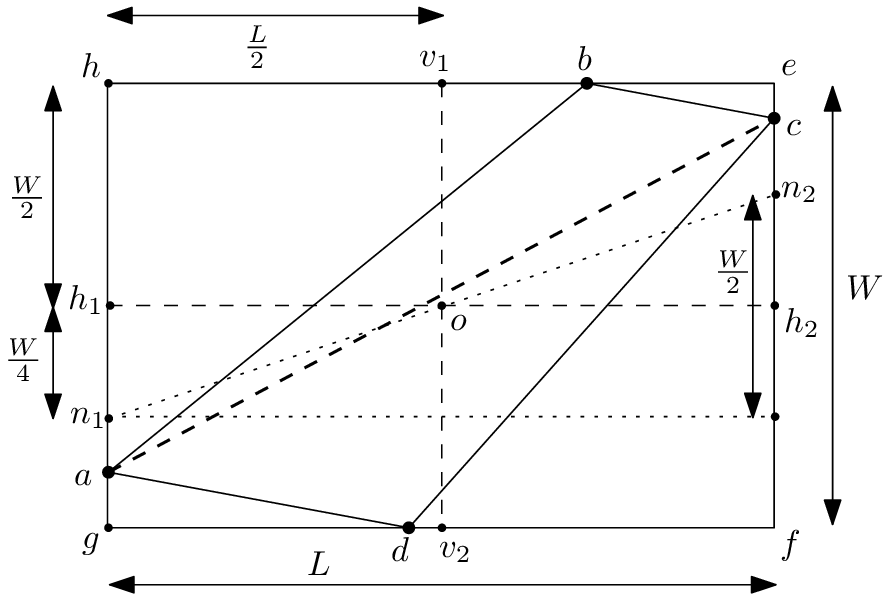}
\caption{Proof of Case $B.1$}
 \label{sm6a}
\end{minipage} %
\hspace{0.5cm}
\begin{minipage}[b]{0.5\linewidth}
\centering
\includegraphics[scale=0.8]{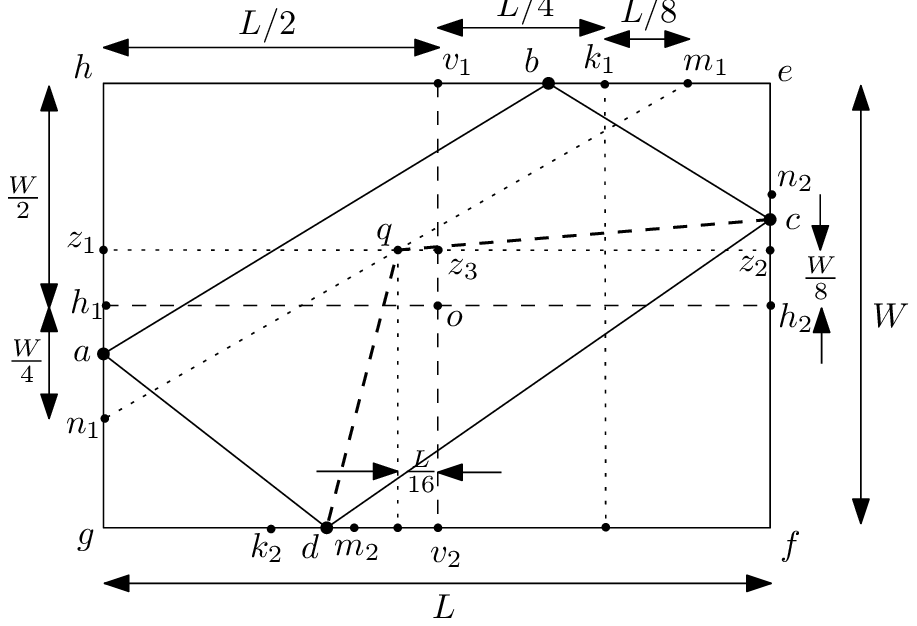}
\caption{Proof of Case $B.2.a$}
\label{sm7a}\end{minipage} %
\end{figure}

{\bf Case B.2:~ $|y(c)-y(a)|< \frac{W}{2}$}

In this case, if $c\in \overline{en_2}$ then $a\notin \overline{n_1g}$; 
similarly, if $a\in \overline{n_1g}$ then $c\notin \overline{en_1}$. Henceforth, 
without loss of generality, we assume that $a \in \overline{h_1n_1}$ (see 
Figure~\ref{sm7a} and Figure~\ref{sm8}(a \& b)). Let $k_1$ and $k_2$ be the mid-points of $\overline{v_1e}$
and $\overline{gv_2}$ respectively.
Take two points $m_1\in \overline{k_1e}$ and $m_2 \in \overline{k_2v_2}$ such 
that $|{k_1m_1}|=\frac{|k_1e|}{2}=\frac{L}{8}$ and $|{k_2m_2}|=\frac{|k_2v_2|}{2}=\frac{L}{8}$. Now depending on 
the position of $b$ on the line segment $\overline{v_1e}$, we divide this case 
into two sub-cases as follows:
 
{\bf Case {B.2.a} : ~ $b\in \overline{v_1m_1}$}

Here $|v_1b|\leq |v_1m_1|=\frac{3L}{8}$ (see Figure~\ref{sm7a}). Connect $m_1$ with $n_1$. 
Consider a horizontal line segment $\overline{z_1z_2}$ above
$\overline{h_1h_2}$ at a distance of $\frac{W}{8}$
which intersect $\overline{m_1n_1}$ and $\overline{v_1v_2}$ at the points $q$ and $z_3$ respectively.
From the similar triangles $\triangle n_1z_1q$ and $\triangle n_1hm_1$, we have 
$\frac{|n_1z_1|}{|z_1q|}=\frac{|n_1h|}{|hm_1|}$.
This gives $\frac{\left(\frac{W}{4}+\frac{W}{8}\right)}{|z_1q|}=
\frac{\left(\frac{W}{4}+\frac{W}{2}\right)}
{\left(\frac{L}{2}+\frac{3L}{8}\right)}$.
Hence, $|z_1q|=\frac{7L}{16}$. Therefore, $|qz_3|=\left(\frac{L}{2}-\frac{7L}{16}\right)=\frac{L}{16}$.
Now, whatever be the position of $a\in \overline{n_1h_1}$,
and $b\in \overline{v_1m_1}$, the point $q$, obtained  
above, must lie inside the $\Diamond abcd$.
Thus we can choose  $\triangle qcd$
whose three sides are of $|cd|\geq \sqrt{\left(\frac{L}{2}\right)^2+\left(\frac{W}{2}\right)^2}$,
$|cq|\geq |z_2q|=|z_2z_3|+|qz_3|=\left(\frac{L}{2}+\frac{L}{16}\right)
=\frac{9L}{16}$
and $|dq|\geq \left(\frac{W}{2}+\frac{W}{8}\right)=\frac{5W}{8}$.
Therefore, the smallest side $\ell$ of $\triangle qcd$ may be $\overline{cq}$, or $\overline{dq}$.
If $\ell=|{cq}|\geq\frac{9L}{16}$, then $\alpha=\frac{r}{(9L/32)}=\frac{8}{9}\sqrt{1+4(W/L)^2}\leq\frac{8}{9}\sqrt{1+4}<2$ 
(since  $\frac{\sqrt{3}}{2}<\frac{W}{L}\leq 1$).

\begin{figure}[htbp]
\begin{minipage}[b]{0.5\linewidth}
\centering
\includegraphics[scale=0.8]{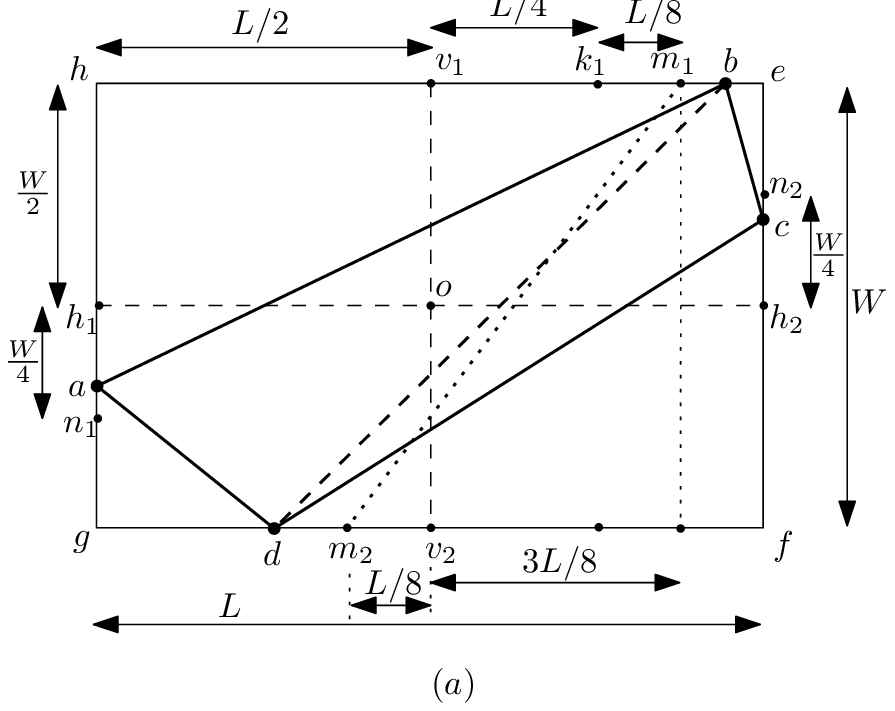}
%\caption{Proof of case $B.2.b(i)$}
%\label{sm8b}
\end{minipage} %
\hspace{0.5cm}
\begin{minipage}[b]{0.5\linewidth}
\centering
\includegraphics[scale=0.8]{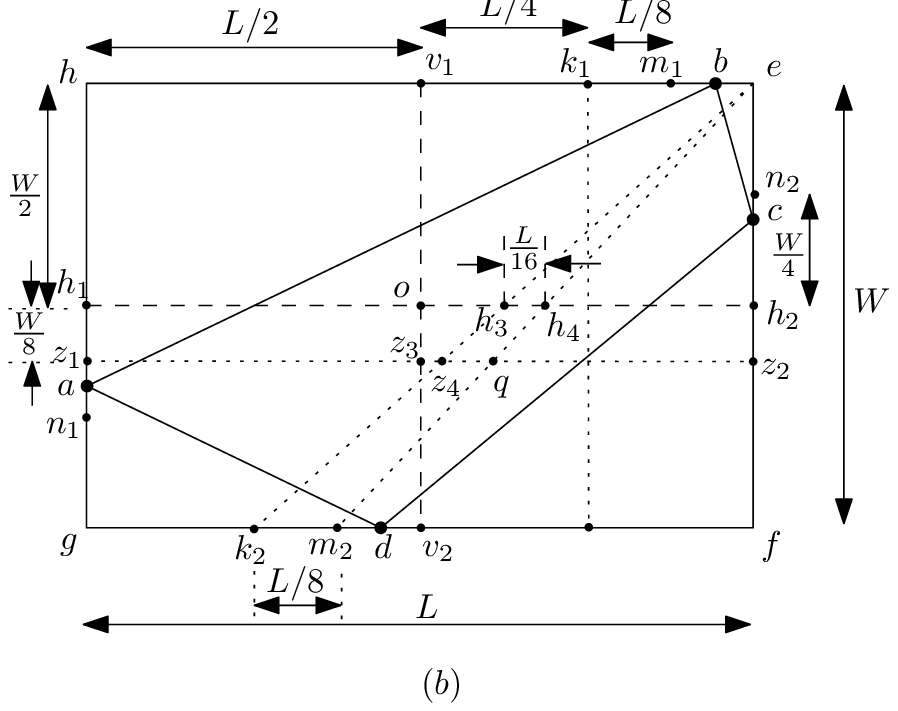}
%\caption{Proof of case $B.2.b(ii)$}
%\label{sm8a}
\end{minipage} %
\caption{(a) Proof of case $B.2.b$(i), and (b) Proof of case $B.2.b$(ii)}
\label{sm8}
\end{figure}

On the other hand, if $\ell=|{dq}|\geq\frac{5W}{8}$,  then  $\alpha=\frac{r}{5W/16}=\frac{4}{5}\sqrt{4+(L/W)^2}\leq\frac{4}{5}\sqrt{4+(4/3)}$
(since $\frac{L}{W}<\frac{2}{\sqrt{3}}$). Thus we have  $\alpha=\frac{16}{5\sqrt{3}}<2$.

{\bf Case {B.2.b}:~ $b\in \overline{m_1e}$ }

Here $|v_1b|>|v_1m_1|=\frac{3L}{8}$ (see 
Figure~\ref{sm8}). Depending on the position of $d$ on $\overline{gv_2}$, we divide this case into two sub-cases as follows:

{\bf Case {B.2.b(i)}:~ $d \in \overline{gm_2}$}

Refer to Figure~\ref{sm8}(a). In this case, the horizontal distance between $b$ and $d$ is 
given by $|x(b)-x(d)|\geq |x(m_1)-x(m_2)|=\frac{L}{2}$, because $b$ lies to the right of $m_1$ and $d$ lies
to the left of $m_2$.  
Thus $|{bd}|\geq |m_1m_2|=\sqrt{W^2+(L/2)^2}$ (see Figure~\ref{sm8}(a)). Therefore, the longest
line segment $\cal L$ inside the $\Diamond abcd$ is at least $|{bd}|$, and hence
we have lower bound $\rho={\cal L}/4\geq |{bd}|/4$. Thus, the approximation factor $\alpha$ will be given by 
$\alpha=\frac{r}{{\cal L}/4}\leq\sqrt{\frac{L^2+4W^2}{W^2+(L/2)^2}}=2$

{\bf Case {B.2.b(ii)}: ~ $d \in \overline{m_2v_2}$}

Refer to Figure~\ref{sm8}(b), where $k_1$ and $k_2$ are the mid-points of $\overline{v_1e}$ and $\overline{gv_2}$ respectively.
In this case, the ``horizontal distance'' between $b$ and $d$ ($|x(b)-x(d)|$) may be less than $\frac{L}{2}$ 
(see Figure~\ref{sm8}(b)). If this ``horizontal distance'' is greater than or equal to $\frac{L}{2}$, we
can show that $\alpha\leq 2$ following the aforesaid ``Case B.2.b(i)''. Hence we study the case, when this
``horizontal distance'' is less than $\frac{L}{2}$.
Connect $m_2$ with $e$ and $k_2$ with $e$ by dotted lines.
  Consider a horizontal line segment $\overline{z_1z_2}$ below $\overline{h_1h_2}$ at a distance of $\frac{W}{8}$. This segment 
 $\overline{z_1z_2}$ intersect $\overline{v_1v_2}$, $\overline{k_2e}$ and $\overline{m_2e}$ at
 the points $z_3$, $z_4$ and $q$ respectively (see Figure~\ref{sm8}(b)). Now $|k_2m_2|=\frac{L}{8}$.
 Here, $\overline{h_1h_2}$ bisects both $\overline{k_2e}$ and $\overline{m_2e}$ at the points $h_3$ and $h_4$ respectively.
 Therefore, $|h_3h_4|=\frac{1}{2}|k_2m_2|=\frac{L}{16}$. Hence,
 $|z_3q|>|z_4q|>|h_3h_4|=\frac{L}{16}$. The point $q$ determined by
 the aforesaid way must lie always inside the quadrilateral
 $\Diamond abcd$ because $d\in \overline{m_2v_2}$ cannot lie to the left of $m_2$
 and $c(\in \overline{eh_2})$ cannot lie above $e$.
 Therefore, we can always choose the triangle $\triangle abq$ whose three side 
 are of length $|ab|\geq\sqrt{(W/2)^2+(7L/8)^2}$, $|aq|\geq |z_1q|=\frac{L}{2}+\frac{L}{16}=\frac{9L}{16}$
 and $|bq|\geq \frac{W}{2}+\frac{W}{8}=\frac{5W}{8}$.
 Hence, the smallest side of the $\triangle abq$ will be
 either $\overline{aq}(\geq \frac{9L}{16})$ or $\overline{bq} (\geq \frac{5W}{8})$. As in 
 {Case~$B.2.a$}, we can show here that the approximation factor  $\alpha \leq 2$.
 %\end{itemize}

 Observation~\ref{approx_1} and the above case analysis suggests the following result. The exhaustiveness of the case analysis is justified with the flowchart in Figure~\ref{flowchart_stream}. Here at each branch point (shown by $\bigcirc$), the branches considered are exhaustive. 
 
\begin{lemma}\label{lll1}
If the subpolygon $\Pi$ is a quadrilateral $\Diamond{abcd}$, then  $\alpha$ is upper bounded by 2. 
\end{lemma}
\begin{figure}
\centering
\includegraphics[scale=0.65]{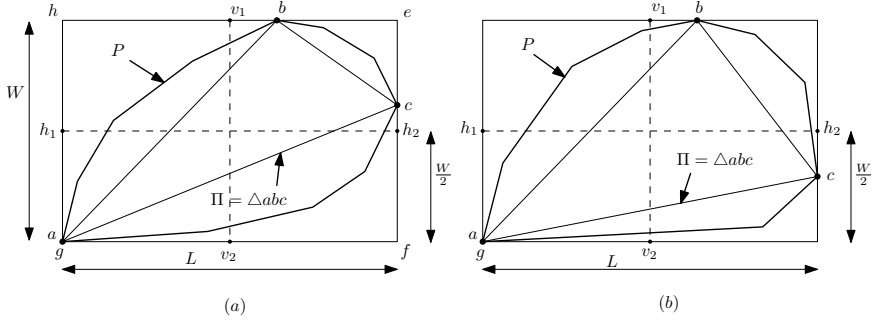}
\caption{In $\triangle abc$ (a) the vertex $c\in \overline{ef}$ lies above $\overline{h_1h_2}$,
and  (b) the vertex $c\in \overline{ef}$ lies below $\overline{h_1h_2}$}
\label{sm_triangle_1}
\end{figure}

\vspace{-0.1in}
\subsubsection{$\Pi$ is a triangle $\triangle abc$}

If only one vertex of $P$ coincides with a vertex of its covering rectangle ${\cal R}=\Box efgh$,
the subpolygon $\Pi$ ($\Pi\subseteq P$) will be a triangle, say $\triangle abc$ (see Figure~\ref{sm_triangle_1}).
Without loss of generality, let us name that vertex of $P$ as ``$a$''
which coincides with ``$g$'' of ${\cal R}=\Box efgh$. Note that $b\in \overline{he}$ and
$c\in \overline{ef}$ (see Figure~\ref{sm_triangle_1}). Here we consider two possibilities depending on
the position of the vertex $c$.
\vspace{-0.1in}
\subsubsection*{(i) $c$ lies above $\overline{h_1h_2}$}
\vspace{-0.1in}
This is shown in Figure~\ref{sm_triangle_1}(a). Here, $|ac|\geq \sqrt{L^2+(W/2)^2}$. Now, the longest line segment
$\cal L$ inside this triangle $\triangle abc$ will be at least of length $|ac|$. Hence, the approximation factor $\alpha$
is given by $\alpha\leq\frac{r}{{\cal L}/4}=\sqrt{\frac{L^2+4W^2}{L^2+(W/2)^2}}\leq 2$, since $W\leq L$.
\vspace{-0.1in}
\subsubsection*{(ii) $c$ lies below $\overline{h_1h_2}$}
\vspace{-0.1in}
This is shown in Figure~\ref{sm_triangle_1}(b). Here,
both ``$a$'' and ``$c$'' lies below $\overline{h_1h_2}$. Hence, by Observation~\ref{obsh1h2} the approximation factor $\alpha$ is $2$.

\begin{lemma}\label{lll2}
If the subpolygon $\Pi$ {\em($\Pi\subseteq P$)} is a triangle 
$\Delta{abc}$, then $\alpha$ is upper bounded by~2. 
\end{lemma}

\vspace{-0.1in}
\subsubsection{$\Pi$ is a diagonal of ${\cal R}=\Box efgh$}

If two vertices of a given convex polygon $P$ coincide with two non-adjacent vertices (say $e$ and $g$ as shown
in Figure~\ref{straight_line}) of ${\cal R}=\Box efgh$, we get its subpolygon  $\Pi$ as a diagonal $\overline{eg}$ of $\cal R$.
The longest line segment ${\cal L}=|eg|=\sqrt{L^2+W^2}$. Therefore the
approximation factor $\alpha$ is given by
$\alpha\leq\frac{r}{{\cal L}/4}=\sqrt{\frac{L^2+4W^2}{L^2+W^2}}\leq\sqrt{1+\frac{3}{1+(L/W)^2}}\leq 2$, since $(W\leq L)$.

\begin{figure}
\centering
\includegraphics[scale=0.65]{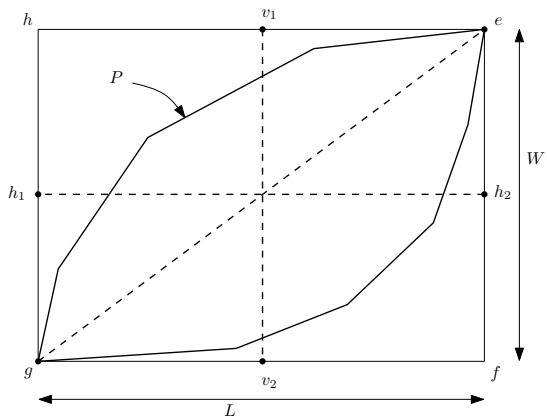}
\caption{The diagonal $\overline{eg}$ of $P$ is exactly covered by the rectangle ${\cal R}=\Box efgh$}
\label{straight_line}
\end{figure}

\begin{lemma}\label{lll3}
If the subpolygon $\Pi$ {\em ($\Pi\subseteq P$)} is a diagonal of
the covering rectangle $\cal R$, 
then $\alpha$ is upper bounded by~2. 
\end{lemma}

\tikzstyle{startstop} = [rectangle, rounded corners, minimum width=2cm, minimum height=1cm,text centered, draw=black, fill=red!30]

\tikzstyle{io} = [trapezium, trapezium left angle=70, trapezium right angle=110, minimum width=0.1cm, minimum height=1cm, text centered, draw=black, fill=blue!30]

\tikzstyle{process} = [rectangle, minimum width=2.5cm, minimum height=1cm, text centered, draw=black, fill=yellow!30]

\tikzstyle{connection} = [circle, text centered, draw=black, fill=yellow!30]

\tikzstyle{decision} = [diamond, minimum width=2cm, minimum height=1cm, text centered, draw=black, fill=green!30]

\tikzstyle{arrow} = [thick,->,>=stealth]
\begin{figure}
\centering
\resizebox{1.15\textwidth}{!}{%
\begin{tikzpicture}[node distance=2cm]

%\node (start) [startstop] {Start};

\node (pro0) [process] {Convex Polygon $P$};
\node (start) [startstop,xshift=-1cm, left of=pro0] {Start};
\node (pro1) [process, right of=pro0, xshift=3cm, text width= 5cm] {Compute $a,b,c~\& ~d$ with min-x, max-y, max-x \text{and} \\ min-y co-ordinates respectively};

\node (pro2) [process, right of=pro1, xshift=4cm, text width= 4cm] {$\Diamond abcd$ is covered by axis-parallel rectangle $\Box efgh$
of length $L$ and width $W$};

\node (pro3) [process, below of=pro1, xshift=0cm, text width= 6cm] {$h_1$, $v_1$, $h_2$ and $v_2$ are the
mid-points of  $\overline{gh}$, $\overline{he}$, $\overline{ef}$ \& $\overline{fg}$};

\node (pro4) [process, left of=pro3, xshift=-5cm, text width=6cm]{The ratio $(W/L)\leq (\sqrt{3}/2)$ and \\ 
{\bf $ \alpha\leq 2$} \textcolor{blue}{({\bf Case I})}};

\node (pro5) [process, below of=pro4,yshift=-2cm, xshift=-2cm, text width=5cm]{$a$ and $c$ lie at the same
side of $\overline{h_1h_2}$ and \\ {\bf $\alpha\leq 2$} \textcolor{blue}{({\bf Observation 5})}};

\node (pro6) [process, right of=pro5, xshift=4cm, text width=5cm]{$b$ and $d$ lie at the same side of
$\overline{v_1v_2}$  and \\ {\bf $\alpha\leq 2$}
\textcolor{blue}{({\bf Observation 6})}};

\node (pro7) [process, right of=pro6, xshift=5cm, text width=8cm]{$a$ \& $c$ lie on 
the opposite side of $\overline{h_1h_2}$, and $b$ \& $d$ lie on the opposite side of $\overline{v_1v_2}$\\
{\bf Assumption:} 
$a$ lies on $\overline{gh_1}$ whereas $c$ lies on $\overline{eh_2}$};

\node (con6) [connection, radius=0.1cm, below of=pro7,xshift=0cm,yshift=-0.01em]{};

%\node (pro8) [process, below of=pro7,yshift=-0.6cm, text width=4cm]{Assumption: 
%$a$ lies on $\overline{gh_1}$ whereas $c$ lies on $\overline{eh_2}$};

\node (pro9) [process, left of=con6,xshift=-4cm,yshift=-2cm, text width=4cm]{$b$ 
lies on $\overline{hv_1}$ and $d$ lies on $\overline{v_2f}$\\ \textcolor{blue}{\bf Case II.A}};

\node (con2) [connection, radius=0.1cm, left of=pro9,xshift=-3.5cm,yshift=-2em]{};

\node (pro10) [process, left of=pro9, xshift=-6cm, yshift=-2.5cm, text width=4cm]{$|{cd}|
\geq \sqrt{\left(\frac{L}{4}\right)^2+\left(\frac{W}{2}\right)^2}$ \\
{\bf $\alpha\leq 2$} \textcolor{blue}{(\bf Case A.1)}};

\node (pro11) [process, right of=pro10,xshift=3cm, text width=4cm]{$|{cd}|< \sqrt{\left(\frac{L}{4}\right)^2+\left(\frac{W}{2}\right)^2}$ \\
{\bf $\alpha\leq 2$} \textcolor{blue}{(\bf Case A.2})};

\node (pro12) [process, right of=pro9, xshift=4cm, text width=4cm]{$b$ lies on $\overline{v_1e}$ 
and $d$ lies on $\overline{gv_2}$\\
\textcolor{blue}{(\bf Case II.B})};

\node (pro14) [process, below of=pro12,xshift=-4cm,yshift=-2cm, text width=4cm]{$|y(a)-y(c)|<(W/2)$\\
\textcolor{blue}{(\bf Case B.2)}};

\node (pro13) [process, right of=pro14, xshift=3cm, text width=4cm]{$|y(a)-y(c)|\geq(W/2)$\\ 
{\bf $\alpha\leq 2$} \textcolor{blue}{(\bf Case B.1)}};

\node (con1) [connection, radius=0.1cm, below of=pro3,yshift=-0.5cm]{};

\node (con3) [connection, radius=0.1cm, below of=pro12,xshift=0cm,yshift=0em]{};

\node (con4) [connection, radius=0.1cm, below of=pro14,xshift=0cm,yshift=-0.001em]{};

\node (pro15) [process, below of=con4,xshift=-2.5cm,yshift=-0.001cm,text width=4cm]{$|v_1b|\leq(3L/8)$\\
{\bf $\alpha\leq 2$}
\textcolor{blue}{(\bf Case B.2.a)}};

\node (pro16) [process, below of=con4,xshift=2.5cm,yshift=-0.001cm, text width=4cm]{$|v_1b|>(3L/8)$\\

\textcolor{blue}{(\bf Case B.2.b)}};

\node (con5) [connection, radius=0.1cm, below of=pro16,xshift=0cm,yshift=-0.01em]{};

\node (pro17) [process, below of=con5,xshift=-3cm,yshift=-0.1em, text width=4.3cm]{$|x(b)-x(d)|\geq(L/2)$\\
{\bf $\alpha\leq 2$}
\textcolor{blue}{(\bf Case B.2.b(i))}};

\node (pro18) [process, below of=con5,xshift=3cm,yshift=-0.1em, text width=4.3cm]{$|x(b)-x(d)|<(L/2)$\\
{\bf $\alpha\leq 2$}
\textcolor{blue}{(\bf Case B.2.b(ii))}};

%\node (stop) [startstop, below of=pro3,xshift=-3cm,yshift=-22em] {Stop};

\draw [arrow] (start) -- (pro0);
\draw [arrow] (pro0) -- (pro1);
\draw [arrow] (pro1) -- (pro2);
\draw [arrow] (pro2) |- (pro3);
\draw [arrow] (pro3) -- (pro4);
\draw [arrow] (pro3) -- node[xshift=1.7em]{\textcolor{blue}{{\bf Case II: $(W/L)>(\sqrt{3}/2)$}}}(con1);
\draw [arrow] (pro7) -- (con6);
\draw [arrow] (con6) -| (pro9);
\draw [arrow] (con6) -- (pro12);
\draw [arrow] (pro9) -| (con2);
\draw [arrow] (con2) -| (pro10);
\draw [arrow] (con2) -| (pro11);
\draw [arrow] (con1) -| (pro5);
\draw [arrow] (con1) -| (pro6);
\draw [arrow] (con1) -| (pro7);
%\draw [arrow] (pro2) |- node[xshift=3em,yshift=1em] {Yes} (con1);
\draw [arrow] (pro12) -- (con3);
\draw [arrow]  (con3) -| (pro13);
\draw [arrow]  (con3) -| (pro14);
\draw [arrow] (pro14) -- (con4);
\draw [arrow] (con4) -| (pro15);
\draw [arrow] (con4) -| (pro16);
\draw [arrow] (pro16) -- (con5);
\draw [arrow] (con5) -| (pro17);
\draw [arrow] (con5) -| (pro18);

\end{tikzpicture}
}
\caption{Flowchart of case study in Streaming Model}
\label{flowchart_stream}
\end{figure}
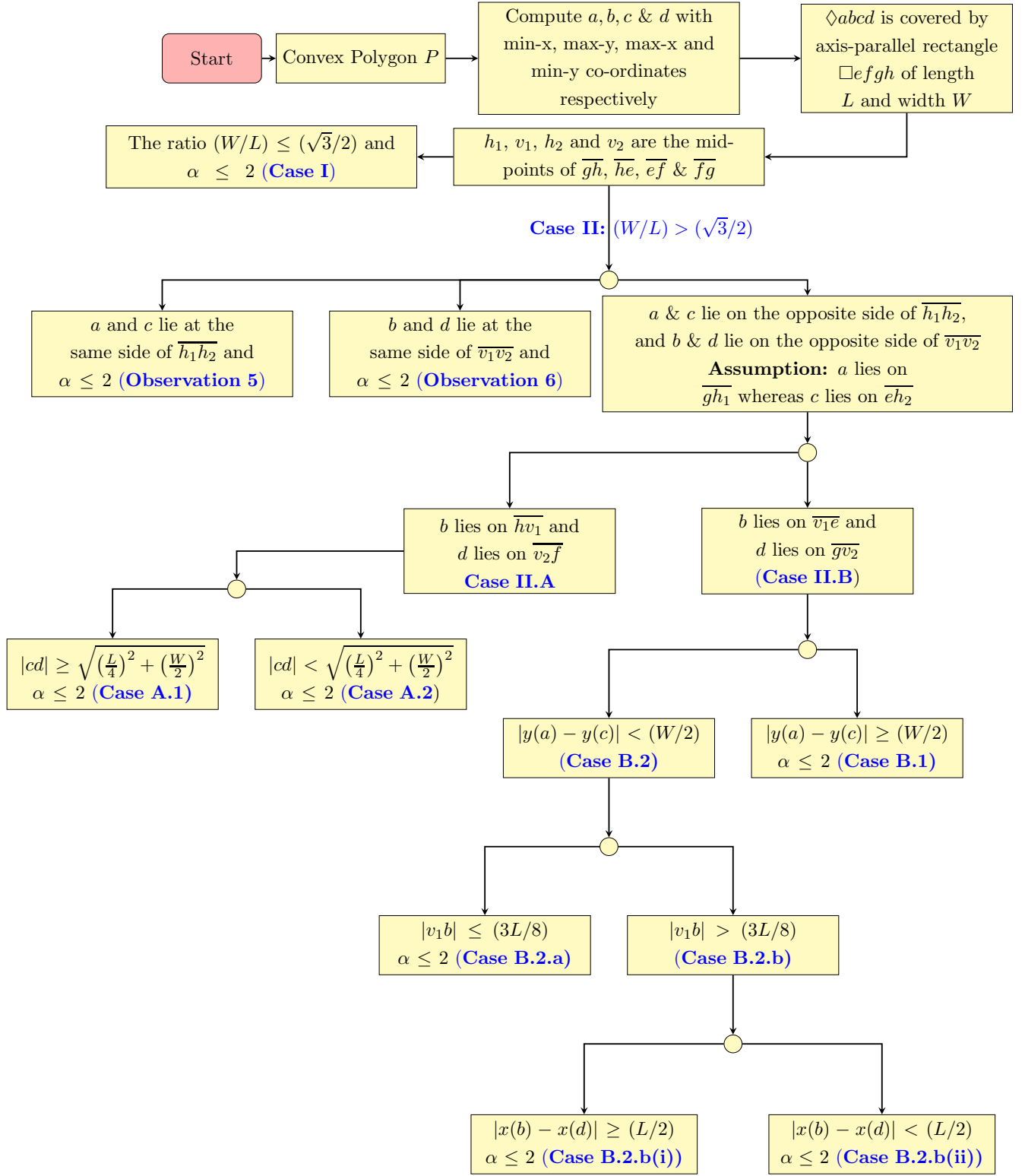

\vspace{-0.1in}
\section{Two-center problem for convex polygon under non-streaming model}
\label{non_str}
In this section, we show that if the computational model is relaxed to non-streaming,
then a simple linear time algorithm can produce a solution with improved approximation
factor of $1.86$. We assume that the vertices of the input polygon $P$
is stored in an array in order. The algorithm and the analysis of approximation factor
are discussed in the following subsections.

\subsection{Proposed algorithm}

We compute the diameter of P. Next, we rotate the coordinate axis around its origin such that the diameter
of the given polygon $P$ becomes parallel to the $x$-axis. 
We use $D$ to denote the length of $D$. 
Let $\cal R$ be an axis-parallel rectangle of length $D$ 
and width $W$ that exactly covers $P$. 
Split $\cal R$ into two equal parts ${\cal R}_1$ and ${\cal R}_2$ by a 
vertical line. Finally, compute two  congruent disks $C_1$ and 
$C_2$ of minimum radii, say $r$, circumscribing ${\cal R}_1$ and ${\cal R}_2$ 
respectively. We report the radius $r$, and the centers of $C_1$ and $C_2$, 
as the output of the algorithm. The time complexity of our 
algorithm is determined by the time complexity of 
computing $D$. Note that, the diameter of the polygon $P$ 
corresponds to a pair of antipodal vertices of $P$ which 
are farthest apart \cite{shamos}. The farthest antipodal 
pair of vertices can be computed by scanning the vertices 
twice in order, and hence it needs $O(n)$ time.

As in the earlier subsection, the approximation factor of 
our algorithm is $\alpha=\frac{r}{r_{opt}}\leq \frac{r}{\rho}$, where $r$ is 
the radius reported by our algorithm and $\rho$ is the lower bound of $r_{opt}$.

\subsection{Analysis of the approximation factor}  
\vspace{-0.1in}

  In this case also, the given polygon $P$ is exactly covered by an axis-parallel 
 rectangle $\mathcal{R}$ having length $D$ (the diameter of the polygon $P$) and 
 width $W$ ($0<W\leq D$), and $r$, the radius of the two enclosing 
congruent disks $C_1$ and $C_2$ computed by our algorithm, is obtained by 
Equation~\ref{eq1}, except that $L$ should be replaced by $D$ (as shown in 
Figure~\ref{rectangle}). Hence $r$ will be given by

 $$r=\sqrt{\left(\frac{D}{4}\right)^2+\left(\frac{W}{2}\right)^2} 
 =\frac{1}{4}\sqrt{D^2+4W^2}.$$

 Without loss of generality, we assume that $D=1$ and $0\leq W \leq  1$. 
Thus, $r=\frac{1}{4}\sqrt{1+4W^2}$.

Now, the 
approximation factor for the polygon $P$ is given by $\alpha
\leq\frac{r}{\rho}$, where $\rho$ is the lower bound of $r_{opt}$.
The lower bound $\rho$ for the problem in this model will also be the 
same as that of used in ``streaming data model'' (discussed in Section~\ref{lower_bound}).
 We may have so many different polygons of 
diameter $D=1$ inside the rectangle $\cal R$, and $\alpha$ may also vary 
depending on the value of $\rho$ for the corresponding polygons. Thus, 
in order to have a better estimate of the upper bound for the 
approximation factor $\alpha$, at first we fix $r$ (or in other 
words both $W$ and $D$ of the rectangle $\cal R$) like in streaming setup.
From the Observation~\ref{approx_1}, we know
that the approximation factor $\alpha$ for two center problem
of a convex polygon $P$ is less than (or equal to) that of its subpolygon $\Pi$
where both $P$ and $\Pi$ are ``exactly covered'' by the rectangle $\cal R$.
Here, the minimal\footnote{The subpolygon $\Pi$ of $P$ is said to be minimal
if no other subpolygon of $\Pi$ is exactly covered by the  
$\cal R$.} subpolygon $\Pi$ will be a quadrilateral which in the degenerate case
may be a triangle.
Now, to have an worst case of $\alpha$, we consider that quadrilateral inside $\cal R$ 
for which $\rho$ (the lower bound of $r_{opt}$) is minimum among  all possible quadrilaterals 
inside ${\cal R}$.

\begin{figure}[htbp]\vspace{-0.1in}
\begin{minipage}[b]{0.3\linewidth}
\centering
 \includegraphics[height=1.1in]{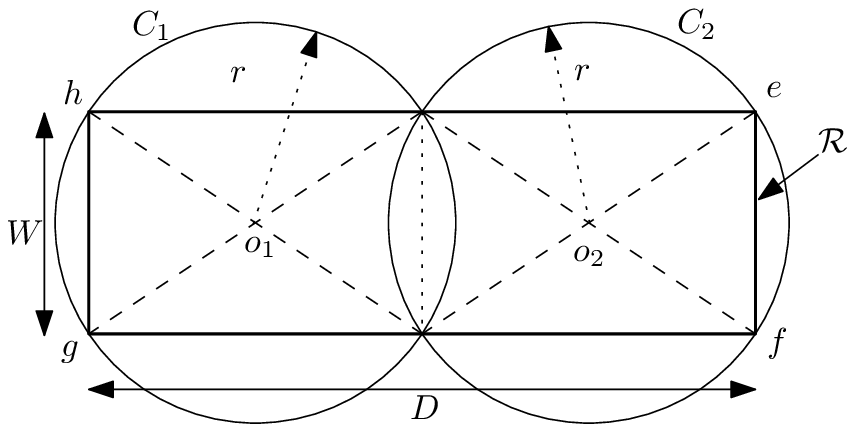}
    \caption{Rectangle $\cal R$ of size $D \times W$  
    covered by two disks $C_1$ and $C_2$}
    \label{rectangle}
    \end{minipage} %
\hspace{0.5cm}
\begin{minipage}[b]{0.6\linewidth}
\centering
	\includegraphics[width=0.7\linewidth]{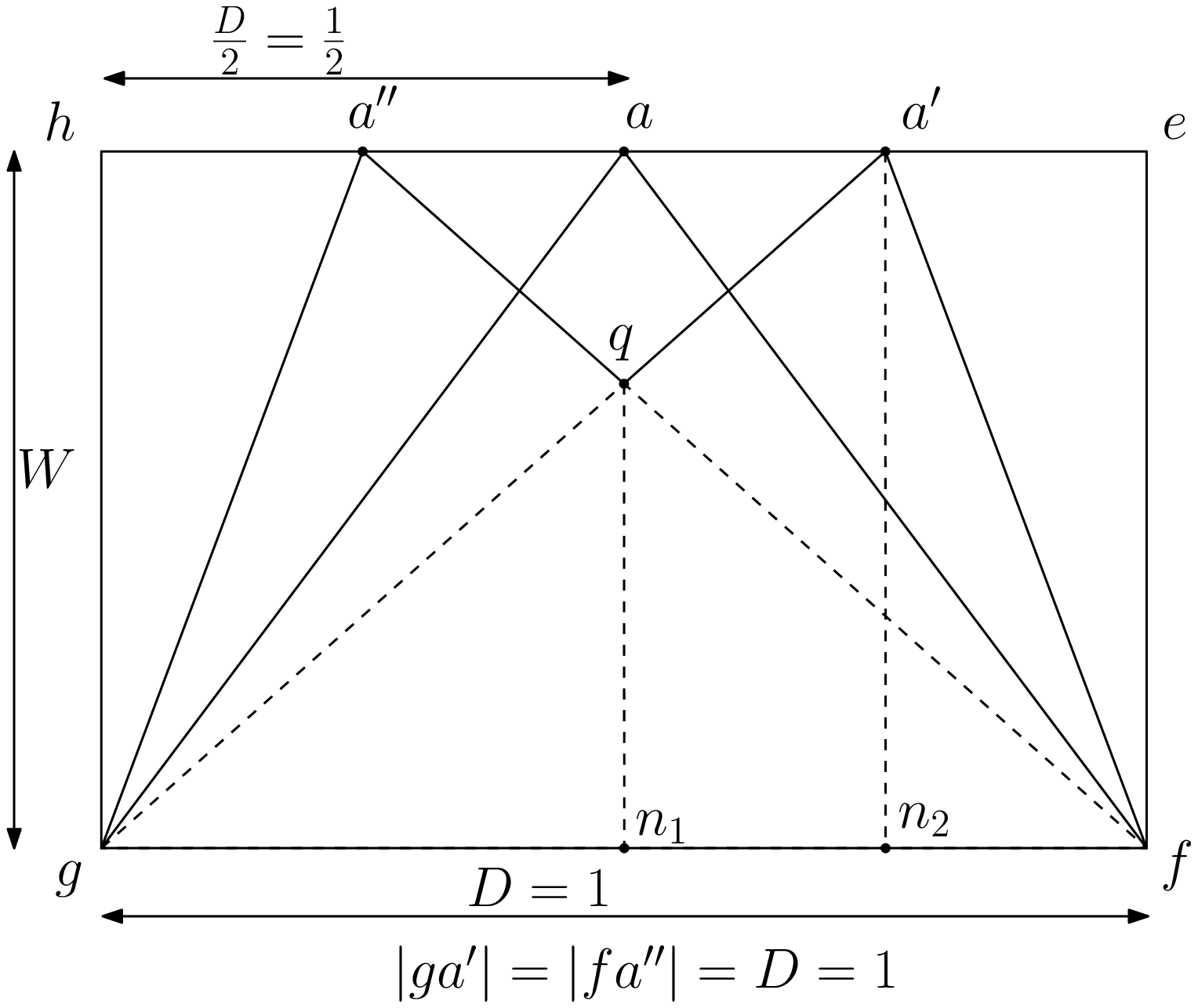}
     \caption{$\triangle gaf$ 
    of diameter $D=1$ covered by  $\Box  efgh$ of length $D=1$ and width $W$.}
    \label{isosceles_triangle}
\end{minipage} %
\end{figure}

In the following subsections, we consider, separately, triangle and quadrilateral as the subpolygon $\Pi$ 
whose approximation factor will give the upper bound for the radius of the two-center
problem of any convex polygon (as discussed in streaming setup). 
Throughout the paper, we always take the diameter $D=1$ and hence, 
the width $W$ of the covering rectangle satisfies $0\leq W \leq 1$. 

    \vspace{-0.1in} 
  \subsubsection{$\Pi$ is a triangle $\triangle gaf$}\vspace{-0.1in}

   Refer to Figure~\ref{isosceles_triangle}. For a convex polygon $P$ which is exactly covered by an isothetic 
   rectangle ${\cal R}=\Box efgh$, we get its subpolygon $\Pi$ as a triangle
   $\triangle gaf$ ($\triangle gaf \subseteq P$) when the diameter of
   the polygon $P$ aligns with an edge, say $\overline{fg}$, of the rectangle 
   ${\cal R}$. In this case, the width $W$ of the covering rectangle ${\cal R} = \Box efgh$ 
   can be at most $\frac{\sqrt{3}}{2}$ (since otherwise $|ga|$ or $|fa|$
   of $\triangle gaf$ will become greater than $1$). We take two points $a'$
   and $a''$ on the edge $\overline{he}$ of $\Box efgh$ so that $|ga'|=|fa''|=|gf|=1$. 
   Note that, the feasible region for $a$ on $\overline{he}$
   is given by $x(a'')\leq x(a)\leq x(a')$. Let $q$ be the point determined by the 
   intersection of $\overline{ga'}$ and $\overline{fa''}$. The isosceles triangle
   $\triangle gqf$ always lies inside $\triangle gaf$
   (see Figure~\ref{isosceles_triangle}) for any position of $a$ on its feasible region. 
   %We consider the isosceles triangle $\triangle gqf$ such that  $\ell \geq |gq|$.
   For an extreme position $a'$ of $a$, the triangle $\triangle ga'f$ has the two of its sides as: $|ga'|=|gf|=1$. Now $|a'n_2|=W$, where $n_2$ is the projection of 
      $a'$ on the edge $\overline{fg}$. Hence, $|gn_2|=\sqrt{1-W^2}$. Now, take a perpendicular $\overline{qn_1}$ from $q$ on the edge $\overline{fg}$.
   From the similar triangles $\triangle gqn_1$ and $ga'n_2$, we have $\frac{|gq|}{|ga'|}
   =\frac{|gn_1|}{|gn_2|}$. Hence, 
   $|gq|=\frac{1}{2\sqrt{1-W^2}}$.       
    Therefore, inside
    $\triangle ga'f$, we have a triangle $\triangle gqf$
    with its smallest side $\overline{gq}$ and hence, for $\triangle ga'f$, 
    $\ell\geq~|gq|$. Hence, 
   $\alpha\leq\frac{r}{(\ell/2)}=\sqrt{(4W^2+1)(1-W^2)}$
   which becomes maximum for $W=\sqrt{\frac{3}{8}}$, and maximum value of 
   $\alpha=\frac{5}{4}=1.25$.
 \vspace{-0.1in} 
  \subsubsection{$\Pi$ is a quadrilateral $\Diamond abcd$}
Let  $\Diamond abcd$ (of diameter $D=|ac|=1$) be covered 
by a rectangle $\Box  efgh$ whose longest side $\overline{fg}$ ($|fg|=1$) 
is parallel to the diameter $\overline{ac}$ of $\Diamond  abcd$ (see Figure  
\ref{fig_isosceles_a}). We assume
that the diameter of $P$ is parallel to the coordinate axes, i.e., $y(a)=y(c)$. The width of 
$\Box  efgh$ is $|ef|=W$, where $W\leq 1$. Throughout this section, 
we use the following notation:

The points $v_1$ and $v_2$ denote the mid-points of the edges $\overline{eh}$
and $\overline{fg}$ respectively. Similarly, the points $h_1$ and $h_2$ are 
the mid-points of the edges $\overline{gh}$ and $\overline{ef}$ respectively. The vertices 
$a$, $b$, $c$ and $d$ of $\Diamond  abcd$ always lie on 
$\overline{gh}$, $\overline{he}$, $\overline{ef}$ and
$\overline{fg}$ respectively. We will study the 
properties of such a rectangle $\Box  efgh$ by considering the 
two cases: (i) $0 < W\leq\frac{2}{\sqrt{5}}$ and (ii) 
$\frac{2}{\sqrt{5}} < W \leq 1$ separately. The reason for 
choosing $W=\frac{2}{\sqrt{5}}$ 
will be explained later.
\vspace{-0.1in}

\subsubsection*{Case I:~ $0 < W\leq\frac{2}{\sqrt{5}}$}
  \begin{figure}[t]
  	\centerline{	\includegraphics[scale=0.7]{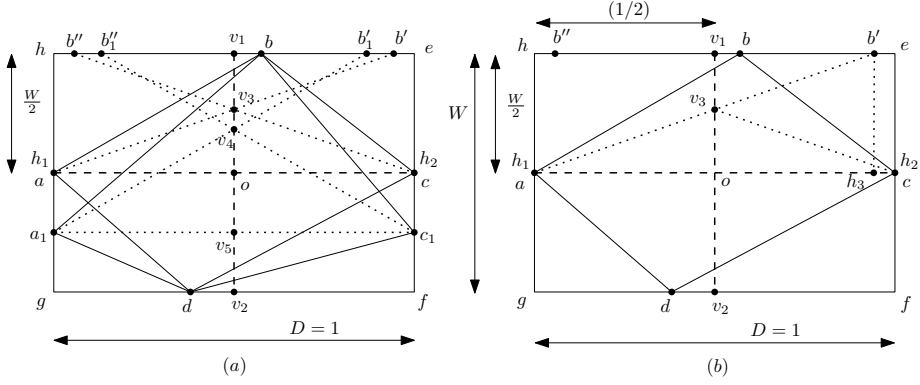}}
  	%\vspace{-1in}
      \caption{(a) Proof of Lemma \ref{lx}, (b) approximation factor in Case I}\vspace{-0.1in}
      \label{fig_isosceles_a}
  	\end{figure} \vspace{-0.1in} 
\begin{lemma} \label{lx}
In a quadrilateral $\Diamond  abcd$, (a) there exists an isocelese triangle 
having base aligned with the diameter of $\Diamond  abcd$, and (b) the other 
two (equal) sides of such a triangle  will have length 
at least $\frac{1}{\sqrt{4-W^2}}$.
\end{lemma}	
\begin{proof}
{\bf Part (a)$\implies$~} 
If the diagonal $\overline{ac}$ of $\Diamond  abcd$ coincides with $\overline{h_1h_2}$ 
(as shown using dark dashed line in Figure~\ref{fig_isosceles_a}(a)), 
then in order to maintain the diameter $D=1$, the feasible region 
of the vertex $b$ of  $\Diamond  abcd$ on the edge $\overline{he}$ of $\Box  
efgh$ is given by $x(b'')\leq x(b)\leq x(b')$, where $b'$ and $b''$
are the points on the edge $\overline{he}$ so that $|ab'|=|cb''|=1$.
In addition, irrespective of the position of 
$b \in \overline{b'b''}$, there exists an isosceles triangle $\triangle 
av_3c$ inside the quadrilateral $\Diamond  abcd$, where $v_3$ is the point
of intersection between $\overline{ab'}$ and $\overline{cb''}$. If the diagonal 
$\overline{ac}$ lies below $\overline{h_1h_2}$ (i.e. for the quadrilateral $\Diamond 
a_1bc_1d$, shown using thin line in Figure~\ref{fig_isosceles_a}(a)), 
then there always exists an isosceles triangle $\triangle a_1v_4c_1$ 
within the quadrilateral $\Diamond a_1b_1c_1d_1$, where $v_4$ is 
determined by the points $b'_1, b_1''\in \overline{eh}$, and the feasible region  
of the vertex $b_1$ which is given by $x(b_1'')\leq x(b_1)\leq x(b_1')$, satisfy the diameter 
constraint $|a_1b'_1|=|c_1b_1''|=1$ of $\Diamond a_1b_1c_1d_1$.

{\bf Part (b) $\implies$~} Let $v_5$ be the mid-point of $\overline{a_1c_1}$. 
Since $|ab'|=|a'b_1'|=1$ and $a_1$ is below $a$ 
on the line $\overline{gh}$, we have $b_1'$ is to the left of $b'$ on the 
line $\overline{eh}$, and the slope of $\overline{ab'}$ is less than that of $\overline{a_1b'_1}$. 
Thus, the portions of $\overline{ab'}$ and $\overline{a_1b'_1}$  between a pair of 
vertical lines $\overline{gh}$ and $\overline{v_1v_2}$ satisfy  $|av_3|<|a_1v_4|$\footnote{$\angle v_3ao < \angle v_4a_1v_5 < 90^o$. $\therefore$
 $\cos\left(\angle v_3ao\right) > \cos\left(\angle v_4a_1v_5\right)$, or $\frac{|ao|}{|av_3|}> \frac{|a_1v_5|}{|a_1v_4|}$. Since, $|ao|=|a_1v_5|$,
 we have $|av_3|<|a_1v_4|$}, where 
$v_3$ and $v_4$ are points of intersection of $\overline{ab'}$ and 
$\overline{a_1b_1'}$ with the vertical line $\overline{v_1v_2}$.

Now, let us consider the quadrilateral $\Diamond  abcd$ as shown 
in Figure~\ref{fig_isosceles_a}(b). Draw perpendiculars $\overline{v_3o}$ 
and $\overline{b'h_3}$ on $\overline{ac}$ from $v_3$ and $b'$ respectively. 
From the similar triangles $\triangle av_3o$ and 
$\triangle ab'h_3$, we have $\frac{|av_3|}{|ao|}=\frac{|ab'|}{|ah_3|}$, 
which gives $\frac{\ell}{1/2}=\frac{1}{\sqrt{1-(W/2)^2}}$, where 
$\ell=|av_3|$. Thus, we have  $\ell=\frac{1}{\sqrt{4-W^2}}$. 
\end{proof}

From Lemma \ref{lx}, we have the approximation factor 
     \vspace{-0.1in}
    \begin{equation}
    \label{eqa}
     \alpha=\frac{r}{(\ell/2)}\leq
    \frac{1}{2}\sqrt{\left(1+4W^2\right)
    \left(4-W^2\right)}
    \end{equation}
    
    Observe that, $\alpha$ is monotonically increasing function in $0\leq W \leq 1$, 
and it attains maximum value for $W=1$, and it is
    $\alpha=\frac{1}{2}\sqrt{\left(1+4\right)\left(4-1\right)}=
    1.936$. Thus, in order to have a smaller approximation factor,
     our objective is to choose a different triangle
    if the width $W$ of the covering rectangle $\Box efgh$ increases 
    beyond a threshold. In Theorem~\ref{theorem3}, we show that
    this threshold is $\frac{2}{\sqrt{5}}$. Thus, in the range 
    $0 \leq W \leq \frac{2}{\sqrt{5}}$, using Equation~\ref{eqa}, we have 
    $\alpha \leq 1.84$.

\vspace{-0.1in}
\subsubsection*{Case II:~ $\frac{2}{\sqrt{5}} < W \leq 1$} 
 
  \begin{observation}
  \label{obs3}
  One of the four sides ($\overline{ab}$, $\overline{bc}$, $\overline{cd}$ and $\overline{db}$) of the quadrilateral $\Diamond  abcd$ must
  be of length at least $\frac{\sqrt{1+W^2}}{2}$.
 \end{observation}
 
 \begin{Proof}
Refer to Figure \ref{observation}. Note that, $|h_1v_1|=|h_2v_1|=|h_1v_2|=|h_2v_2|=\frac{\sqrt{1+W^2}}{2}$.
Thus, if $\overline{ac}$, the diameter of $\Diamond  abcd$ lies on or below $\overline{h_1h_2}$, 
then for any feasible position of the vertex $b$ on the edge $\overline{eh}$, either 
$|ab|$ or $|cb|$ is at least $\frac{\sqrt{1+W^2}}{2}$. 
If $\overline{ac}$ is above 
$\overline{h_1h_2}$ then either $|ad|$ or $|cd|$ is at least $\frac{\sqrt{1+W^2}}{2}$ 
(as shown in Figure \ref{observation}). 
 \end{Proof}
 
\begin{figure}[htbp]\vspace{0.5in}
\begin{minipage}[b]{0.55\linewidth}
\centering
\vspace{-0.3in}
	\includegraphics[width=\linewidth]{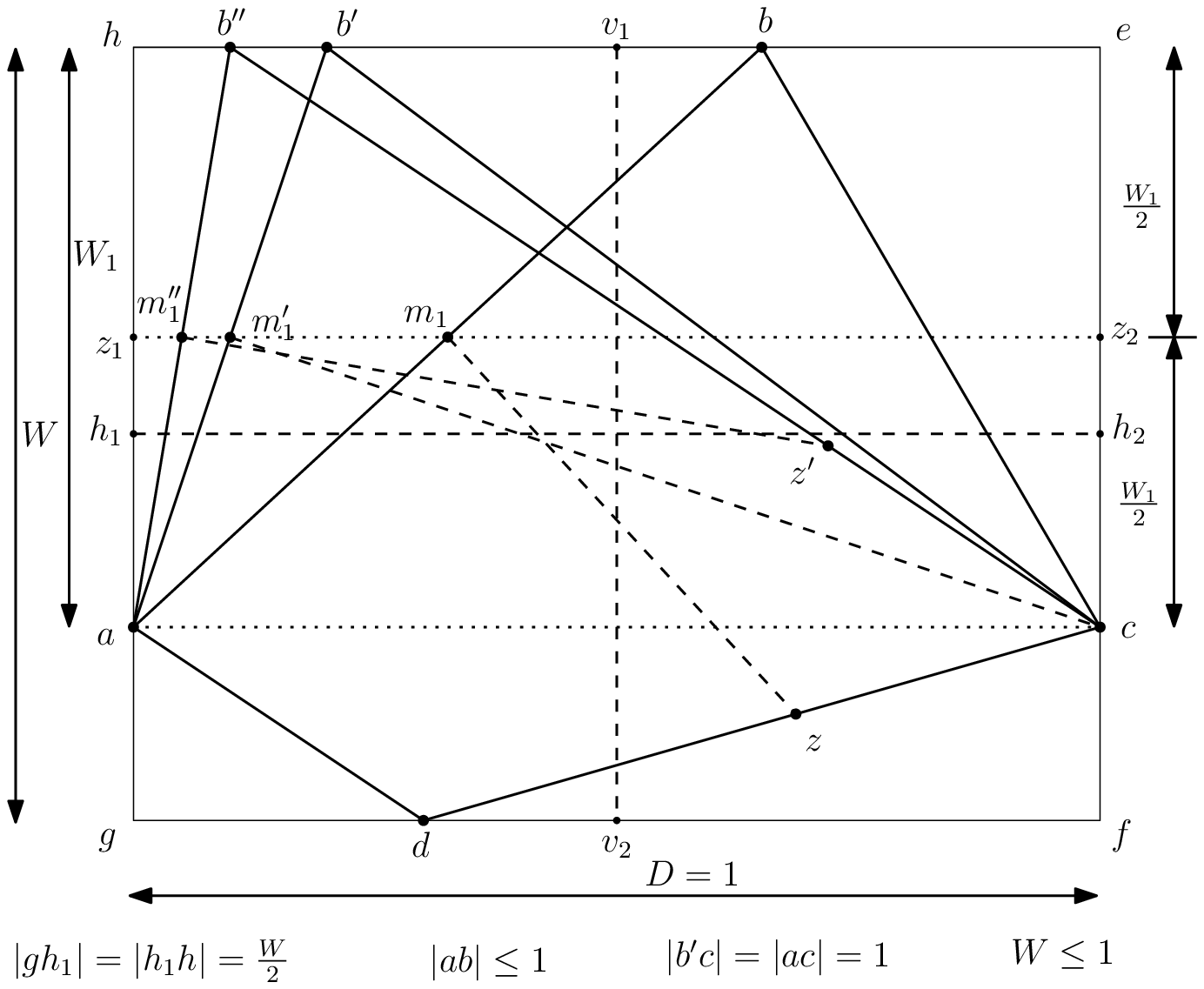}\vspace{-0.1in}
	    \caption{Proof of Observation \ref{obs1}}\vspace{-0.1in}
	    \label{observation}
\end{minipage} %
\hspace{0.5cm}
\begin{minipage}[b]{0.45\linewidth}
\centering
\includegraphics[width=\linewidth]{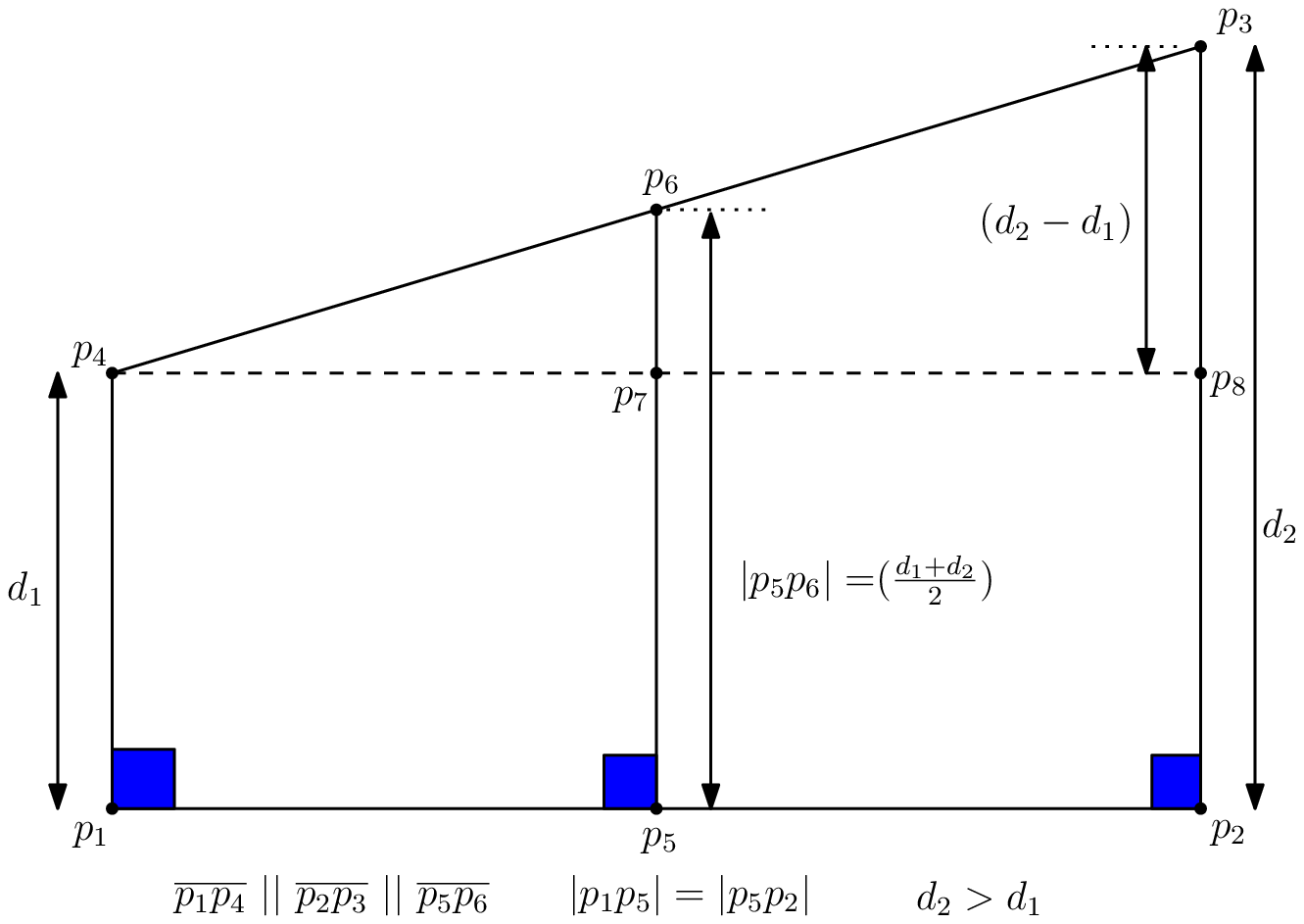}\vspace{-0.1in}
    \caption{The three parallel edges $\overline{p_1p_4}$, $\overline{p_2p_3}$
    and $\overline{p_5p_6}$ in quadrilateral $\Diamond p_1p_2p_3p_4$.}
    \vspace{-0.1in}
    \label{height}
    \end{minipage} %
\vspace{-0.1in}
\end{figure}

Without loss of generality, from now onwards we assume that
$\overline{ac}$ lies below $\overline{h_1h_2}$  and 
the vertex ``$b$'' lies to the right of ``$v_1$'' which makes $|ab|\geq\frac{\sqrt{1+W^2}}{2}$
(following Observation~\ref{obs3}) in $\Diamond  abcd$. The perpendicular 
bisector of the edge $\overline{ab}$ is denoted by $\overline{m_1z}$, where 
$m_1$ is the mid-point of $\overline{ab}$. Now, $\overline{m_1z}$ intersects 
$\Diamond  abcd$ at the point $z$ (see Figure~\ref{observation}).
   
 \begin{lemma}
 \label{lma1}
  If one of the adjacent edges of $\overline{ab}$ is of length at least
  $\frac{\sqrt{1+W^2}}{2}$ and the width $W$ of the covering rectangle 
  $\Box  efgh$ is at least $\frac{1}{\sqrt{3}}$, then 
  $\alpha \leq 1.6$.  
 \end{lemma}
 
\begin{Proof} 
In $\Diamond  abcd$, the two adjacent edges of $\overline{ab}$ are $\overline{bc}$ and $\overline{ad}$.
If $|bc|\geq \frac{\sqrt{1+W^2}}{2}$ then we consider $\triangle abc$, 
and the length of its smallest side $\ell\geq \frac{\sqrt{1+W^2}}{2}$. 
Similarly, if $|ad|\geq \frac{\sqrt{1+W^2}}{2}$ then we consider 
$\triangle abd$ where $|bd| \geq W$. If $W>\frac{1}{\sqrt{3}}$, 
we have $W > \frac{\sqrt{1+W^2}}{2}$ (since $4W^2>1+W^2$). Hence, $\ell$ (the length of 
the smallest side of $\triangle abd$) $\geq$ $\frac{\sqrt{1+W^2}}{2}$. 
Thus in $\Diamond  abcd$, always there exists a triangle whose smallest side 
is of length $\ell \geq \frac{\sqrt{1+W^2}}{2}$. Thus, $\alpha=\frac{r}
{\ell/2}=\sqrt{\frac{1+4W^2}{1+W^2}}=\sqrt{1+\frac{3}{1+\left(\frac{1}
{W^2}\right)}}$. This is a monotonically increasing function of $W$, 
and it attains maximum when $W=1$ to have $\alpha=\sqrt{2.5}<1.6$  
 \end{Proof}
   
   Thus Lemma~\ref{lma1} suggests that, we need to consider the
   case where both the adjacent sides of $\overline{ab}$ 
   are of length strictly  less than $\frac{\sqrt{1+W^2}}{2}$.

 \begin{observation}
 \label{obs1}
  The perpendicular bisector of $\overline{ab}$ (of the quadrilateral
  $\Diamond  abcd$) cannot intersect the edge $\overline{bc}$ except at 
  its end-point $c$.
 \end{observation}\vspace{-0.1in}

  \begin{Proof}
  Refer to Figure~\ref{observation}. In $\Diamond  abcd$, the perpendicular
  bisector ($\overline{m_1z}$) of 
  the edge $\overline{ab}$ intersect $\overline{cd}$ at a point $z$. Now,  if $b$ is moved towards left
  on the edge $\overline{eh}$, the point $z$ on $\overline{cd}$ moves towards $c$. 
  At a position $b'$ (say) of $b$, the perpendicular bisector $\overline{m'_1z}$  of $\overline{ab'}$ passes through $c$. Then 
  $\triangle ab'c$ becomes isosceles with $|b'c|=|ac|=1$. If 
  we try to make $ \overline{m'_1z}$ intersect with $\overline{bc}$, we need to 
  move $b$ to the left of $b'$. For any such point $b''$, we 
  have $|cb''|>1$, violating the diameter constraint of 
  $\Diamond  abcd$.
 \end{Proof}

  \begin{fact}
  \label{fact3}
   In a quadrilateral $\Diamond p_1p_2p_3p_4$, if  $\overline{p_2p_3}$
   and $\overline{p_1p_4}$ are perpendicular to  $\overline{p_1p_2}$,
   and the segment $\overline{p_5p_6}$, touching $\overline{p_1p_2}$ and $\overline{p_3p_4}$, is
   the perpendicular bisector of $\overline{p_1p_2}$ (see Figure~\ref{height}),
   then  $|p_5p_6|=\left(\frac{d_1+d_2}{2}\right)$,
   where $d_1=|p_1p_4|$ and $d_2=|p_2p_3|$. 
  \end{fact}

 \begin{figure}
  \vspace{-0.3in}
	\centering
	\includegraphics[width=\linewidth]{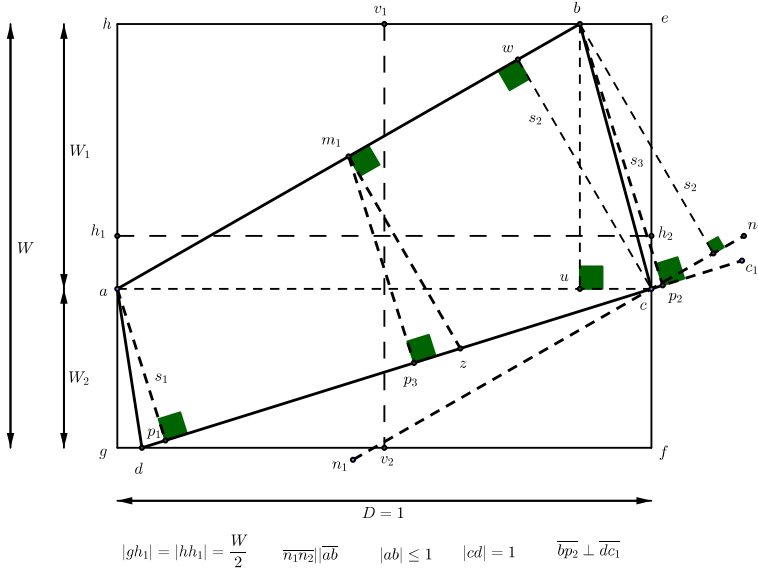}\vspace{-0.1in}
    \caption{Perpendicular bisector $\overline{m_1z}$ of 
  $\overline{ab}$ intersect $\overline{cd}$ at $z$ and
  $| {m_1z}|>\frac{W}{2}$}\vspace{-0.1in}
    \label{proof3}
	\end{figure}
	
 \begin{lemma}
 \label{lm1}
  In the quadrilateral $\Diamond  abcd$, if the perpendicular bisector $\overline{m_1z}$ of 
  $\overline{ab}$ intersects its non-adjacent
  edge $\overline{cd}$ at a point $z$, then $|m_1z|\geq \frac{W}{2}$.
 \end{lemma}

  \begin{Proof}   
  Consider the scenario  where  $\Diamond  abcd$ satisfies the following (Figure~\ref{proof3}): 
  \begin{itemize}
   \item The diagonal $\overline{ac}$ of  $\Diamond  abcd$ is below 
   $\overline{h_1h_2}$ so that
   $|ah|=|ec|=W_1\geq W/2$ and $|ag|=|cf|=W_2\leq W/2$.
   \item The point $d$ is chosen on $\overline{fg}$ such that $|cd|=1$. The point 
   $b$ is chosen at any arbitrary position to the right of $v_1$ on $\overline{eh}$ such that 
   $|ab| \leq 1$.
  \end{itemize}
We show that in such a scenario, $| {m_1z}|\geq W/2$. If we move $d$
   to the right (towards $f$) along the edge $ \overline{fg}$, keeping $a$, $b$, $c$
   fixed, then $|m_1z|$ increases.
   
Let $a$ be the origin of the co-ordinate system, and $|ab|\leq 1$.
  Thus the vertices of the quadrilateral $\Diamond  abcd$ are $b=(\sqrt{|ab|^2-W_1^2},W_1)$, $c=(1,0)$ and $d=((1-\sqrt{1-W_2^2}),-W_2)$. 
   The equation of the lines $\overline{ab}$ and $\overline{cd}$ are given by 
  $y=\frac{W_1}{\sqrt{|ab|^2-W_1^2}}x$ and 
  $y=\frac{W_2}{\sqrt{1-W_2^2}}(x-1)$ respectively. Let the point $p_1$
  be the projection of $a$ on $\overline{cd}$ and the point $w$ be 
  projection of the point $c$ on $\overline{ab}$ (see Figure~\ref{proof3}). Let 
  
  \begin{equation}
  \label{eqn_d1}
   s_1=|ap_1|=\frac{{\frac{W_2}{\sqrt{1-W_2^2}}}}
   {\sqrt{1+\left(\frac{W_2}{\sqrt{1-W_2^2}}\right)^2}}=W_2
  \end{equation}
    
\begin{equation}\label{eq3}
 s_2=|cw|=\frac{{\frac{W_1}{\sqrt{|ab|^2-W_1^2}}}}
   {\sqrt{1+\left(\frac{W_1}{\sqrt{|ab|^2-W_1^2}}\right)^2}}
   =\frac{W_1}{|ab|}\geq W_1 ~~(\text{since} ~a\leq 1)
\end{equation}

Now, if  
$u$ is the projection of $b$ on $\overline{ac}$, we have $|bu|=W_1$.
Thus from Figure~\ref{proof3} we have, $\sin\angle cab=\sin\angle uab=\frac{|bu|}{|ab|}=
\frac{W_1}{| {ab}|}\geq W_1$ 
(since $| {ab}|\leq 1$) and
similarly $\sin\angle acd=\frac{s_1}{|ac|}=\frac{W_2}{|ac|}=W_2$ (since $s_1=W_2$ from Equation~\ref{eqn_d1}, and $|ac|=1$).
Since $W_1 > W_2$, we have $\sin\angle cab>\sin\angle acd$ which implies $\angle cab>\angle acd$.

We now draw a line $\overline{n_1n_2}$ through the point $c$ and parallel to $\overline{ab}$. The 
perpendicular distance of this line from $b$ is  $s_2$. 
The line segment $\overline{dc}$ is extended to $\overline{dc'}$ such that it can contain 
the projection $p_2$ of vertex $b$ on the edge $\overline{dc}$ (or on its extension).  
Now, we consider the two cases:
\begin{itemize}
\item If the projection $p_2$ of $b$ on $\overline{dc}$ is to the left of $c$, then 
 $|bp_2|\geq |bu|=W_1$.
 \item If the projection $p_2$ of $b$ on $\overline{dc}$ is to the right of $c$, then 
 since the slope of $\overline{n_1n_2}$ is greater than that of $\overline{dc'}$, we have 
 $|bp_2|\geq s_2\geq W_1$ (see Equation \ref{eq3}).
\end{itemize}

\vspace{-0.1in}
Let $p_3$ be the projection of 
$m_1$ ($m_1$ is the mid-point of $\overline{ab}$) on the line $\overline{cd}$.
Now, consider the quadrilateral 
$\Diamond abp_2p_1$. Using 
Fact~\ref{fact3}, we have $|m_1p_3|=\frac{s_1+s_2}{2}\geq 
\frac{W_1+W_2}{2}=\frac{W}{2}$ because $s_1=W_2$ and $s_2\geq W_1$ (see Equations~\ref{eqn_d1} and \ref{eq3}).
Note that, $\overline{m_1z}$ is the 
perpendicular bisector of $\overline{ab}$ which meets $\overline{cd}$ at $z$, and 
$\overline{m_1p_3}$ is the perpendicular from $m_1$ on $\overline{cd}$. Thus, 
$|m_1z|\geq|m_1p_3|=\frac{W}{2}$.
 
Thus, we proved that if $|cd|=1$ then $|m_1z|\geq \frac{W}{2}$.
Now, for a fixed position of $a$, $b$ and $c$ we reduce $|cd|$ 
by moving $d$ towards $f$ along the edge $\overline{fg}$. Thus, 
$|m_1z|$ increases further, and hence we have $|m_1z|\geq \frac{W}{2}$ 
at any position of $d$ on $\overline{fg}$.
\end{Proof}
\vspace{-0.1in}
\begin{theorem}
 Always there exist a triangle $\triangle$ within the quadrilateral 
 $\Diamond  abcd$ so that the length of smallest side of $\triangle$
 is at least $\frac{\sqrt{1+5W^2}}{4}$.
\end{theorem}
\vspace{-0.1in}
 \begin{Proof}
 Observation~\ref{obs1} says that the perpendicular
 bisector $\overline{m_1z}$ of $\overline{ab}$ must intersect either $\overline{cd}$ or
 $\overline{ad}$. We consider these two cases separately.\vspace{-0.1in}
 \begin{description}
\item[$\bullet \overline{m_1z}$ intersects $\overline{cd}$:]  
   By Lemma~\ref{lm1},  $|m_1z|>\frac{W}{2}$.
Since $|ab|\geq \frac{\sqrt{1+W^2}}{2}$ (by the 
assumption following Observation~\ref{obs3}), we can choose the isosceles 
triangle $\triangle abz$ having equal sides 
$\overline{az}$ and $\overline{bz}$, and their (common) length satisfies
\begin{equation}
\label{eq4}
 |az|=|bz|\geq \sqrt{\left(\frac{\sqrt{1+W^2}}{4}
 \right)^2+\left(\frac{W}{2}\right)^2}=\frac{\sqrt{1+5W^2}}{4}
\end{equation}
As $W\leq 1$, we have  $\frac{\sqrt{1+5W^2}}{4} < 
\frac{\sqrt{1+W^2}}{2}$. Thus $\ell$, the length of the 
smallest side of $\triangle abz$ is at least 
$\frac{\sqrt{1+5W^2}}{4}$.
\item[$\bullet \overline{m_1z}$ intersects $\overline{ad}$:]
Consider the extension of the perpendicular
bisector $\overline{m_1z}$ (of $\overline{ab}$) that intersects $\overline{fg}$
at $d_0$ (see Figure~\ref{intersect}). Thus,
if the vertex $d$ of 
$\Diamond  abcd$ coincides with $d_0$, then $\overline{m_1z}$ will touch both
$\overline{cd}$ and $\overline{ad}$,  and in that case 
$|m_1z|=|m_1d_0|$. Now, 
by Lemma~\ref{lm1}, we have $|m_1d_0|=|m_1z|>\frac{W}{2}$, 
and as in Equation~\ref{eq4}, we have $|ad_0|
=\sqrt{|am_1|^2+|m_1d_0|^2}
\geq \frac{\sqrt{1+5W^2}}{4}$. In this case,
we obtain an isosceles $\triangle abd_0$ and the length of its smallest
side satisfy $\ell=|ad_0|=|bd_0| \geq\sqrt{\left(1+5W^2\right)}/4$.

 However, if $d$ lies to the right of $d_0$, say at $d_1$ (see dashed 
 lines in Figure~\ref{intersect}),
 we  consider  $\triangle abd_1$, and we have 
 $|ad_1|>|ad_0|>\frac{\sqrt{1+5W^2}}{4}$ 
 and $|bd_1|>~W>~\frac{\sqrt{1+W^2}}{2}$ (as $W>\frac{2}{\sqrt{5}}$ implies $W>\frac{1}{\sqrt{3}}$, i.e. $4W^2>1+W^2$). Now
 $\frac{\sqrt{1+5W^2}}{4} < \frac{\sqrt{1+W^2}}{2}$
 for $W<\sqrt{3}$ which is obvious because $W\leq D=1<\sqrt{3}$. Therefore, in this case also the length of the smallest
 side ($\ell$) of a triangle satisfy  $\ell \geq \frac{\sqrt{1+5W^2}}{4}$.
 \end{description}
 \vspace{-0.2in}
\end{Proof} 

\begin{figure}[t]
  \vspace{-0.3in}
	\centering
	\includegraphics[scale=0.65]{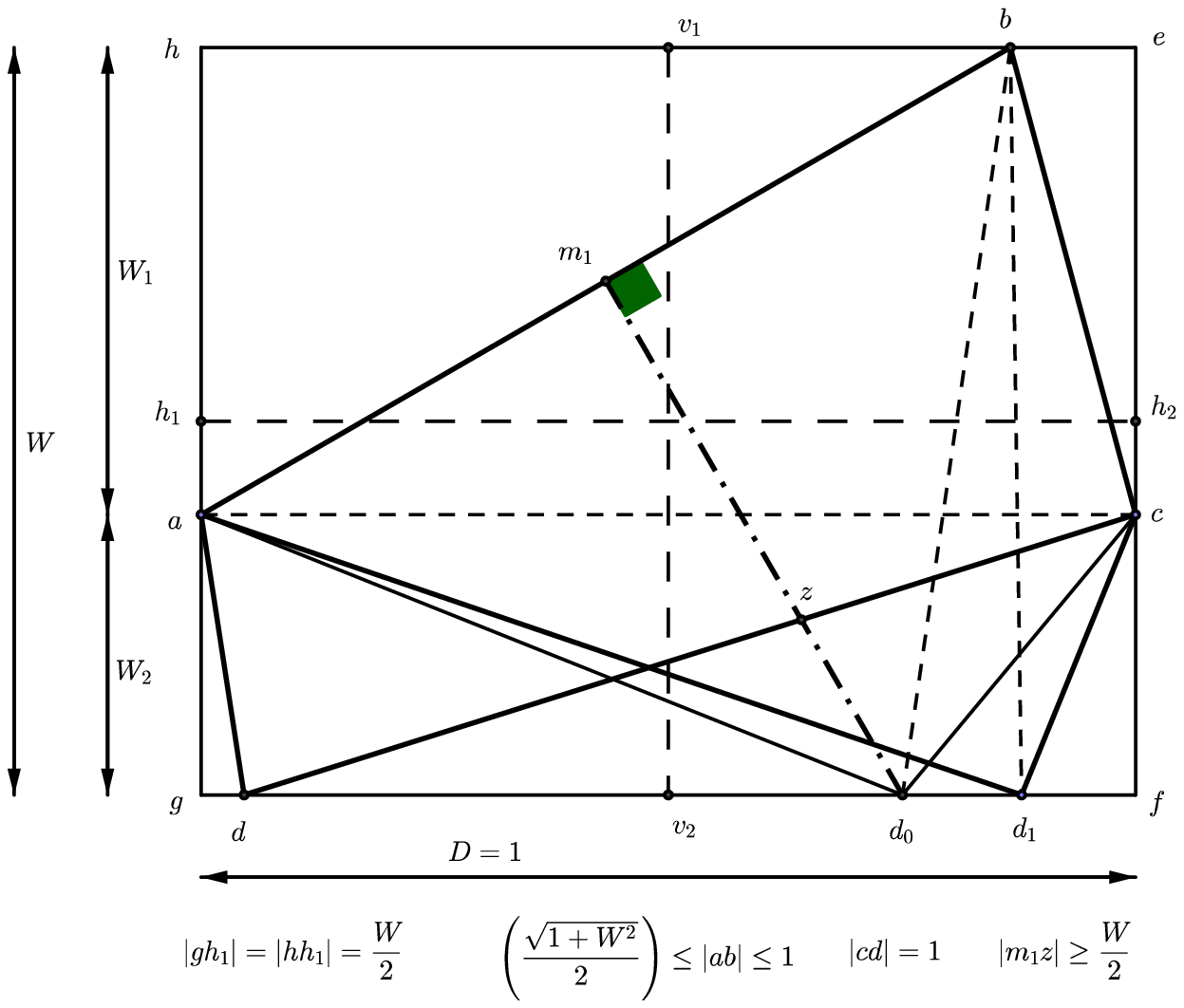}
    \caption{The perpendicular bisector of $\overline{ab}$ can
    intersect the edge $\overline{ad}$ of $\Diamond  ~abcd$ if $d$ lies at the
    right of $d_0$}
    \label{intersect}
	\end{figure}

Thus we have the approximation factor  
\begin{equation}
\label{eqb}
 \alpha=\frac{r}{\ell/2}=\frac{\frac{\sqrt{1+4W^2}}{4}}{\frac{\sqrt{1+5W^2}}{8}}
 =2\sqrt{\frac{1+4W^2}{1+5W^2}}=2\sqrt{1-\frac{1}{\left(5+\frac{1}{W^2}\right)}}, 
\end{equation}
which  is a decreasing function in $W$.

\begin{lemma}
\label{theorem3}
 The approximation factor $\alpha$ for two-center problem of a quadrilateral $\Diamond abcd$ is given by
 $\alpha <1.84$.
\end{lemma}
\begin{Proof}
 The increasing and decreasing nature of the value of $\alpha$ with respect to $W$ in 
 Equations~(\ref{eqa}) and (\ref{eqb}), respectively suggest a threshold value
of $W$ based on which we decide which triangle is to 
be selected inside the quadrilateral $\Diamond abcd$.
Equating the expressions of $\alpha$ in Equations~(\ref{eqa}) and (\ref{eqb}), we have 
\begin{equation}
\frac{\sqrt{(1+W^2)(4-W^2)}}{2}=2\sqrt{\frac{4W^2+1}{5W^2+1}}\\
\Rightarrow (5W^2-4)(W^2-3)(4W^2+1)=0
\end{equation}
The only feasible solution of the above equation is $W={2}/{\sqrt{5}}=0.8944$.
The approximation factor $\alpha$ at $W=\frac{2}{\sqrt{5}}$ using Equation \ref{eqb} is given by
 $\alpha=1.833<1.84$. 
Thus the result in the stated theorem is justified as follows:
\begin{itemize}
 \item[(i)] {\bf $W < \frac{2}{\sqrt{5}}$ :} 
 Choose the isosceles triangle $\triangle acv_3$ 
 with its smallest side $\ell=| {av_3}|$ 
as in Figure~\ref{fig_isosceles_a}.
\item[(ii)] {$W\geq \frac{2}{\sqrt{5}}$ :} based on the intersection 
between $\overline{cd}$ and the perpendicular bisector $\overline{m_1z}$
of $\overline{ab}$, the 
following two sub-cases occur:
  \begin{enumerate}
   \item[(a)] {\bf  $\overline{m_1z}$ intersect $\overline{cd}$ at $z$:} Choose
   the isosceles triangle $\triangle abz$ with its smallest side 
   $\ell\geq \frac{\sqrt{1+5W^2}}{4}$.
   \item[(b)] {\bf $\overline{m_1z}$ intersect
    $\overline{ad}$ :} Choose
   the triangle $\triangle abd$ with its smallest side 
    $\ell\geq \frac{\sqrt{1+5W^2}}{4}$.
  \end{enumerate}
\end{itemize}
\vspace{-0.2in}
\end{Proof}
\vspace{-0.2in}

\begin{theorem}
\label{theorem4}
 The approximation factor $\alpha$ for two-center problem of a given convex polygon
 $P$ is given by $\alpha <1.84$.
\end{theorem}
\begin{Proof}
Observation~\ref{approx_1} says that the approximation factor $\alpha$
 for two-center problem of a given convex polygon $P$ 
 is less than (or equal to) that of its minimal subpolygon $\Pi=\Diamond abcd$,
 where both $P$ and $\Pi$ are exactly covered by the rectangle $\cal R$.
 Now, the result follows from Lemma~\ref{theorem3}.
\end{Proof}

\subsubsection{Special Case: $W=D=1$}
  We now show one special case when the 
     covering rectangle {\cal R} is a square i.e., $W=D=1$ (see Figure~\ref{comparison}). 
     This will give us
     an idea that the upper bound of the approximation factor of our algorithm cannot be
     smaller than $1.527$.
  Any quadrilateral inscribed within this ``square $efgh$''
  must be of a diamond shape (i.e. the $x-$coordinate of
  two points $b$ and $d$ must be equal).
  There are two extreme situations: one
  with $|ab'|=1$, where the corresponding quadrilateral is 
  $\Diamond   ab'cd'$ and the other one is a 
  square $\Diamond  abcd$ (Figure~\ref{comparison}) respectively.\\
  
   {\bf (i) For quadrilateral $\Diamond  abcd$}: 
  
  Refer to Figure~\ref{comparison}. Since $|ab'|=|ad'|=|b'd'|=|ac|=1$,
  we have  $\angle cab'=30^{o}$.  Now $\angle cab=45^{o}$ because $\Diamond abcd$ is a square. 
  The points $t_1$ and $t_2$ are the points of intersection of $\overline{ab'}$ with $\overline{bc}$, and  
  $\overline{ad'}$ with $\overline{cd}$ respectively
Therefore, $\frac{| {co_2}|}{| {o_2t_1}|}=1$
and $\frac{|{ao_2}|}{| {o_2t_1}|}=\sqrt{3}$. 
Now $| {ao_2}|+| {co_2}|=1$
gives $| {o_2t_1}|=\frac{1}{\sqrt{3}+1}$. So, length of each side of the equilateral $\triangle at_1t_2$
is given by $|t_1t_2|=2|o_2t_1|=\frac{2}{\sqrt{3}+1}$.
Therefore, the equilateral
triangle $\triangle at_1t_2$ inscribed
within quadrilateral $\Diamond  abcd$ have the side 
$| {at_1}|=\frac{2}{\sqrt{3}+1}=0.732$,
whereas the isosceles triangle
$\triangle abc$ has the smallest side $| {ab}|=\frac{1}{\sqrt{2}}=0.7071$.
 Thus the smallest side $\ell$
 of the triangle $\triangle at_1t_2$ is the largest inside the $\Diamond abcd$ 
 and we consider $\triangle at_1t_2$ with $\ell=\frac{2}{\sqrt{3}+1}$.
 Hence the approximation factor
$\alpha=\frac{r}{\ell/2}={\frac{(\sqrt{5})}{4}}/{\frac{1}{(\sqrt{3}+1)}}=1.527$.

\begin{figure}[t]
  \vspace{-0.3in}
	\centering
	\includegraphics[scale=0.75]{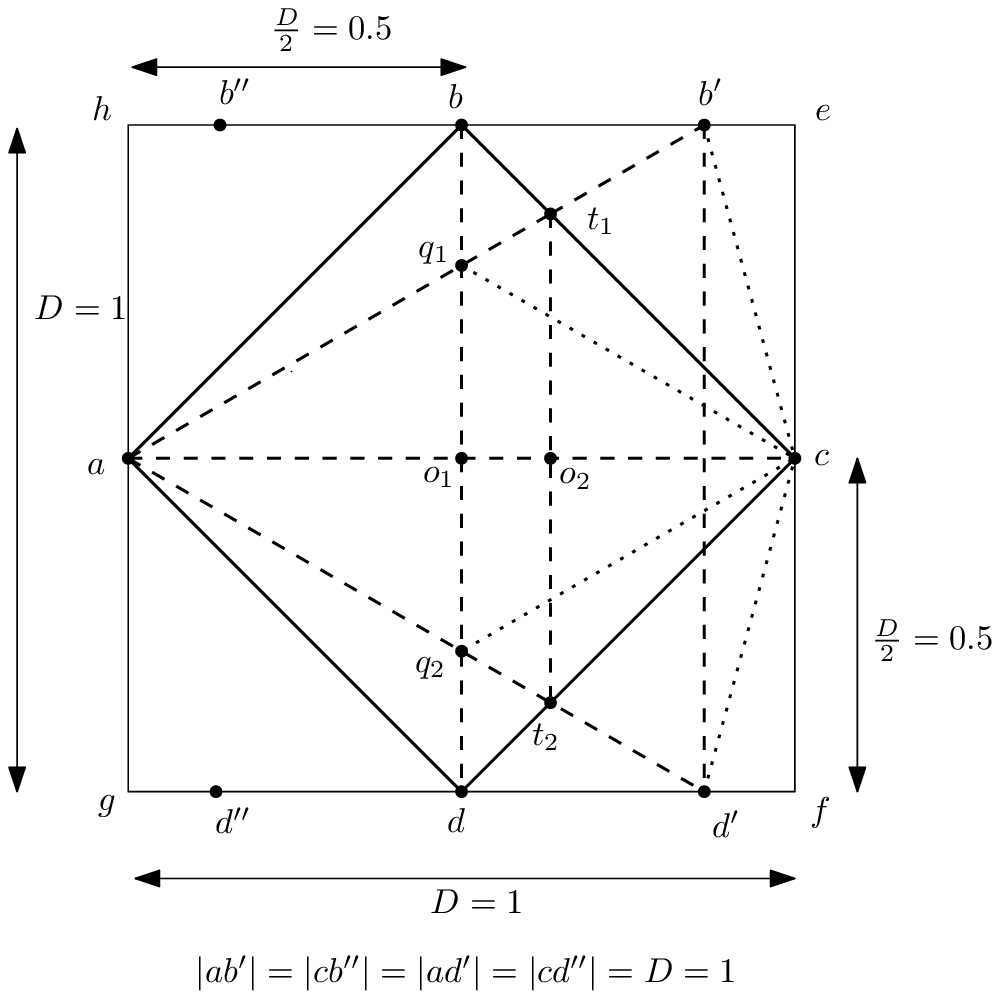}
    \caption{The covering rectangles of the quadrilateral
    $\Diamond  abcd$ is ``square $efgh$''}
    \label{comparison}
	\end{figure}
	
   \remove{
   \begin{wrapfigure}{r}{0.5\textwidth}
	\centering
	\includegraphics[scale=0.65]{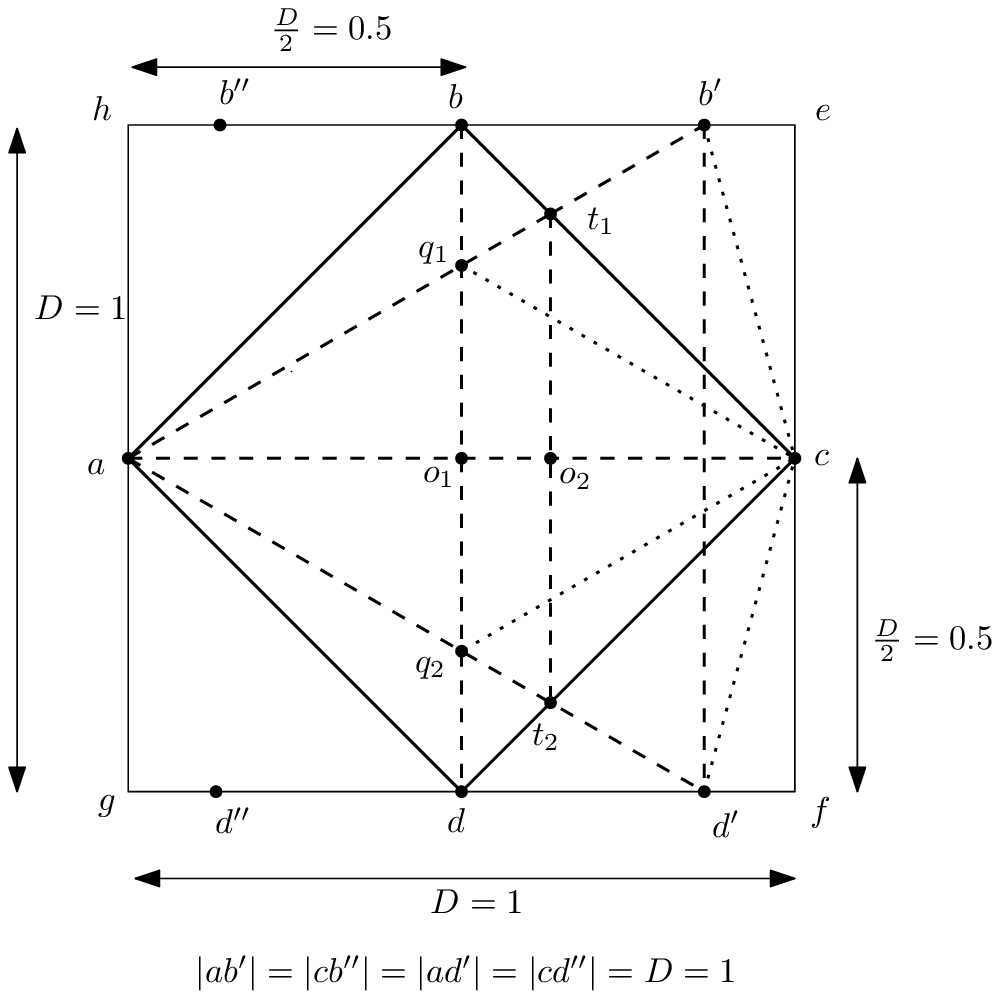}
    \caption{The covering rectangles of the quadrilateral
    $\Diamond  abcd$ is ``square $efgh$''}
    \label{comparison}
  	\end{wrapfigure}
}

 {\bf (ii) For quadrilateral $\Diamond   ab'cd'$}: 

The largest equilateral triangle inscribed within the quadrilateral $\Diamond   ab'cd'$
is $\triangle ab'd'$ (see Figure~\ref{comparison}) whose sides are all $1$.
Thus $\ell=1$, and the approximation factor $\alpha=\frac{r}{\ell/2}$ 
= $\left({\frac{\sqrt{5}}{4}}\right)/\left(\frac{1}{2}\right)=1.118$.

The lower bound of $\ell$ for any quadrilaterals $\Diamond  abcd$ inscribed within the square $efgh$, 
where
the range of $b$ on the edge $\overline{he}$ is given by
$x(b'')\leq x(b)\leq x(b')$, will be the intermediate of the lower bounds for quadrilaterals $\Diamond  abcd$
 and $\Diamond   ab'cd'$.
 Hence if the covering rectangle $\Box efgh$ is a square, the 
 approximation factor $\alpha$ will satisfy 
 $1.118 \leq \alpha \leq 1.527$. 
 This shows that our technique can not produce a solution with
 approximation factor less than
$1.527$, because we need to consider all possible convex polygons for this problem.

 \section{Conclusion and future work}
 \label{conc}
To the best of our knowledge, this is the first work
on approximation for two-center problem of a given convex 
polygon both in streaming and non-streaming setup. In the streaming setup, 
we have designed a $2$-factor approximation 
algorithm using $O(1)$ space for this problem; whereas in the non-streaming setup,
we have proposed a linear time
approximation algorithm with 
approximation factor $1.84$. The ``longest line segment inside a quadrilateral'' and ``the 
triangle which makes its smallest side larger'' have been
considered to determine the lower bound for the radius of the two-center problem of a given convex polygon.

The main bottleneck of adopting the 1.84 factor approximation algorithm in the streaming
model is the unavailability of an algorithm for computing the diameter of a convex polygon
in streaming model. Thus, getting such an algorithm will be an interesting problem to
study.

Surely, improving or establishing non-trivial lower bounds for the approximation results of
this problem will be the main open problems.

\bibliographystyle{splncs03}
\bibliography{research}

\end{document}